\documentclass[a4size, 10pt,final,printed]{IEEEtran}
\usepackage{amsmath}
\usepackage{amssymb}
\usepackage{amsfonts}
\usepackage{graphicx}
\usepackage{epsfig}
\usepackage{subfigure}
\usepackage{psfrag}
\usepackage{algorithm, algorithmic}

\usepackage{booktabs}
\usepackage{ctable}
\usepackage{slashbox}
\usepackage{color}

\newcommand{\thickhline}{\noalign{\hrule height 1pt}}

\title{On Design of Opportunistic Spectrum Access in the Presence of Reactive
Primary Users}

\author{Yue Ling Che,~\IEEEmembership{Student Member,~IEEE}, Rui Zhang,~\IEEEmembership{Member,~IEEE}, and Yi Gong,~\IEEEmembership{Senior Member,~IEEE}
\thanks{Y.~L.~Che and Y.~Gong are with the School of Electrical
and Electronic Engineering, Nanyang Technological University,
Singapore (e-mail: chey0004@ntu.edu.sg; eygong@ntu.edu.sg).}
\thanks{R.~Zhang is with the Department of Electrical and Computer Engineering, National
University of Singapore (e-mail: elezhang@nus.edu.sg). He is also with the Institute for Infocomm
Research, A*STAR, Singapore.}
}

\begin{document}
\maketitle \thispagestyle{empty}

\begin{abstract}
Opportunistic spectrum access (OSA) is a key technique enabling the secondary users (SUs) in a cognitive radio (CR) network to transmit over the ``spectrum holes'' unoccupied by the primary users (PUs). In this paper, we focus on the OSA design in the presence of reactive PUs, where PU's access probability in a given channel is related to SU's past access decisions. We model the channel occupancy of the reactive PU as a 4-state discrete-time Markov chain. We formulate the optimal OSA design for SU throughput
maximization as a constrained finite-horizon partially observable Markov decision process (POMDP)  problem. We solve this problem by first considering the conventional short-term conditional collision probability (SCCP) constraint. 
We then adopt a long-term PU throughput (LPUT) constraint to effectively protect the reactive PU transmission.
We derive the structure of the optimal OSA policy under the LPUT constraint and propose a suboptimal policy with  lower complexity. Numerical results are provided to validate the proposed studies, which reveal some interesting new tradeoffs between SU throughput maximization and PU transmission protection in a practical interaction scenario.
\end{abstract}


\newtheorem{definition}{\underline{Definition}}[section]
\newtheorem{fact}{Fact}
\newtheorem{assumption}{Assumption}
\newtheorem{theorem}{\underline{Theorem}}[section]
\newtheorem{lemma}{\underline{Lemma}}[section]
\newtheorem{corollary}{\underline{Corollary}}[section]
\newtheorem{proposition}{\underline{Proposition}}[section]
\newtheorem{example}{\underline{Example}}[section]
\newtheorem{remark}{\underline{Remark}}[section]
\newcommand{\mv}[1]{\mbox{\boldmath{$ #1 $}}}
\newtheorem{property}{\underline{Property}}[section]

\section{Introduction}
By enabling the secondary users (SUs) to access the unoccupied channels of  the primary users (PUs) in a cognitive radio (CR) network,  opportunistic spectrum access (OSA) is regarded as one promising solution to resolving the spectrum scarcity versus spectrum underutilization paradox in wireless communications \cite{Q.ZhaoDSA}-\cite{Mohanty.radio.06}.
To design  optimal OSA strategies, two competing goals are addressed at the same time:
the  ``spectrum holes'' unused by the PUs  should be optimally explored by the SUs to maximize their throughput, whereas the probability of the SU's transmission collision with undetected active PUs should be minimized.
In this paper, we study the OSA design for SUs in the presence of reactive PUs and aim at achieving the optimal tradeoffs between SU  throughput maximization and  PU  collision minimization.

\subsection{Related Work}
A great deal of valuable prior work has investigated the OSA design for  CR networks.
Assuming that SU is only able to sense a certain part of the spectrum at each time due to hardware limitations, the authors in \cite{Q.ZhaoPOMDP} proposed a partially observable Markov decision process (POMDP) framework to design the optimal OSA.
However, due to ``the curse of dimensionality'', POMDP problems are in general computationally prohibitive to solve \cite{POMDP,POMDP.Survey}.
Zhao~\textit{et~al.} in \cite{Q.ZhaoTWC2008} formulated the design of SU's optimal sensing policy as a POMDP and proposed a myopic sensing policy, which maximizes SU's average reward over a finite horizon.
In \cite{Y.ChenTSP2009}, Chen \textit{et al.} proposed a  threshold-based optimal spectrum sensing and accessing policy, which maximizes SU's throughput during its battery lifetime. Both \cite{Q.ZhaoTWC2008} and \cite{Y.ChenTSP2009} assumed that sensing errors are negligible. In \cite{Y.ChenIT08}, Chen \textit{et al.} considered OSA design in the presence of sensing errors and proposed a short-term conditional collision probability (SCCP) constraint for protecting PUs from SU's collisions in a time-slotted primary system. Moreover, \cite{Y.ChenIT08} proposed a separation principle to significantly reduce the complexity of solving the constrained finite horizon POMDP problem, which  maximizes SU's  throughput subject to the SCCP constraint.
Since the SCCP constraint is able to provide  effective protection to PUs' transmission \cite{Y.ChenIT08}, it has been widely adopted in subsequent studies on OSA. For example, a similar POMDP problem subject to the SCCP constraint was considered in \cite{P.TCISS09} for unslotted PU systems.
An online OSA algorithm by learning PU's  signal statistics was proposed in \cite{UJ.SP10} under the SCCP constraint. Li \textit{et al.} in \cite{Tong.JSAC.2011} showed that when the SCCP constraints over time are tight, the optimal OSA policy can be
implemented as a simple memoryless policy with periodic channel sensing.

Most existing work on OSA with time-slotted PUs, including the aforementioned one, has assumed  a \emph{non-reactive} PU model, where PU's transmission over a particular channel evolves as a 2-state on/off Markov chain with fixed state transition probabilities. Similar assumptions can also be found in the experimental based work on OSA with unslotted PUs, such as  \cite{Wang.Dyspan.05} and \cite{Sadler.CM.07}.
Although greatly simplifying the OSA design, the non-reactive PU model might not be practical since existing wireless systems are mostly intelligent enough to adapt their transmissions upon experiencing collision or interference.
For example, a PU may increase transmit power to compensate the link loss due to the received interference. Alternatively, it may reduce the channel access probability when collision occurs in a carrier sensing multiple access (CSMA) based primary system. In this paper, we refer to such PUs as \emph{reactive} PUs, to differentiate from their non-reactive counterparts.

In this paper, we focus on designing SU's optimal OSA policy in the presence of time-slotted reactive PUs.
It is worth noting that there has been recent work that addressed reactive PUs for OSA and/or spectrum sharing (SS) based CR networks.
In contrast to OSA, with SS, SU is allowed to transmit regardless of the PU's on/off status, provided that the resulting interference to PU is kept below a predefined threshold. 
In \cite{R.ZhangDySpan08}, the authors proposed a \emph{hidden power-feedback loop} for the CR: If PU is reactive and reacts upon receiving SU's interference, SU will receive a power-boosted PU signal that is easier to detect.
Following \cite{R.ZhangDySpan08},  \cite{G.Zhao.TWC.09} proposed a proactive sensing scheme and a sequential transmit power adaptation strategy to exploit spectrum opportunities in the SS based CR.
In \cite{R.ZhangGlobecom09}, the author  extended the work in \cite{R.ZhangDySpan08} and designed active learning and supervised transmission schemes.
Automatic retransmission request (ARQ) based  reactive PUs have been considered in, e.g.,  \cite{Levorato.Allerton.12} and \cite{Levorato.Globecom.10}, for  the SS based CR. Under the assumption that SU has full knowledge of PU's buffer state and  ARQ state, the authors in \cite{Levorato.Allerton.12} adopted a Markov process based model to determine SU's optimal transmission policy over an infinite horizon, which  maximizes SU's long-term average throughput subject to PU's long-term throughput loss. As shown in \cite{Levorato.Allerton.12}, the SU's optimal transmission policy is stationary and thus can be obtained by solving a linear program.
Online algorithms have also been proposed in \cite{Levorato.Globecom.10} for the cases where only partial and/or noisy observations of PU's buffer state and ARQ state are available to the SU.
Compared with the existing work for SS based CR, the work considering reactive PUs for OSA based CR is very limited.
It is noted that a CSMA-based reactive PU model has  been proposed in \cite{R.Chen} to investigate the performance of different SU access policies; however, \cite{R.Chen} did not address the optimal OSA design.

\subsection{Main Results}
In this paper, we focus on the effects of SUs' channel access actions on the reactive PUs' transmission quality. Since the existence of the secondary network is usually oblivious to PUs, we assume that PUs only implement conventional techniques, such as energy detection, to detect the existence of interference/collision; thus, PUs are not able to differentiate the received interference/collision from other PUs and that from SUs. In addition, there might be other unexpected co-channel interference and noise at the primary receiver, which can also evoke reactions of PUs. It is assumed that the reactive PUs treat all the received interference/collision in the same way and react to it accordingly.

We consider an OSA-based CR network, in which one SU transmits opportunistically over $N$ orthogonal frequency bands, each of which is assigned to one PU.
In each time-slot, the SU  selects one channel to sense by choosing a spectrum sensor operating point, and then determines whether to access the selected channel based on the sensing result.
To maximize the SU's throughput subject to PUs' transmission protection, we formulate the OSA design problem as a constrained POMDP problem.
The main results of this paper are summarized as follows.
\begin{itemize}
 \item We propose a new 4-state discrete-time Markov chain model to describe the channel occupancy state  of each reactive PU, which includes the conventional 2-state on/off model for the non-reactive PU as
       a special case.
       The expanded state space and  state transition probabilities in the new model are used to specify the reactions of PU subject to SU's transmit collision.
\item  By adopting the conventional SCCP constraint to protect PU's transmission as in \cite{Y.ChenIT08}, we study the optimal OSA policy under the proposed reactive PU model.
       We extend the \emph{separation principle} proposed in \cite{Y.ChenIT08} for the non-reactive PU case to the reactive PU case, and obtain the optimal OSA policy that can be implemented efficiently.
       However, unlike the non-reactive PU case,  we show that the reactive PU's throughput in general cannot be guaranteed under the SCCP constraint.

\item  To effectively protect the reactive PU's transmission, we adopt a \emph{long-term PU throughput}  (LPUT) constraint, similar to the one proposed for the SS based CR in \cite{R.ZhangGlobecom08}.        Under this constraint,  we first  study  the OSA design for PU's \emph{worst case} transmission with $N=1$, i.e., there is only one pair of PU and SU sharing a single channel.  We obtain the optimal OSA policy structure in this case, which reveals that the spectrum sensor design plays a crucial role in effectively protecting PU's transmission. Noticing the high complexity in designing an effective spectrum sensor due to the non-deterministic belief state transitions of POMDP, we thus  convert the POMDP into an equivalent Markov decision process (MDP) with deterministic state transitions. By studying the reformulated MDP-based LPUT constraint, we propose a suboptimal OSA policy with lower implementation  complexity, which is shown to guarantee the reactive PU's throughput.  Based on the separation principle,
 we then extend the suboptimal policy for the case of $N=1$ to the general case of $N>1$ and show that  the reactive PU's throughput on each channel is  guaranteed by the proposed suboptimal policy.
\end{itemize}

\subsection{Organization}
The rest of this paper is organized as follows. Section II presents the  channel occupancy model for reactive PUs in a CR network.
Section III formulates  the OSA design under the reactive PU model as a constrained POMDP problem.
Section IV studies the POMDP problem under the conventional SCCP constraint and develops the optimal OSA policy based on the separation principle.
Section V studies the POMDP problem under the proposed LPUT constraint and proposes a suboptimal policy.
Section VI  compares numerical examples on
the performance of the proposed optimal and suboptimal policies.
Finally, Section VII concludes the paper.

\section{System Model}
We consider a CR network consisting of one SU and $N$ PUs.
Each PU is preassigned a dedicated channel and the traffic carried by each channel is assumed to be \emph{independent} from each other.
We assume synchronized time-slotted transmission for all the PUs and SU as in \cite{Q.ZhaoPOMDP}, \cite{Q.ZhaoTWC2008}, \cite{Y.ChenTSP2009}, \cite{Y.ChenIT08}, and \cite{UJ.SP10}.
In the following, we model the channel occupancy state of the reactive PU and describe the corresponding OSA decisions of the SU.

\subsection{Channel Occupancy Model for Reactive PU}
Fig.~1 shows a typical example of the channel occupancy  model for the conventional non-reactive PU, which has been adopted in prior work, e.g., \cite{Q.ZhaoDSA}, \cite{Q.ZhaoPOMDP}, \cite{Q.ZhaoTWC2008}-\cite{Tong.JSAC.2011}.
In this model, the primary traffic over a given channel is approximated by a two-state discrete-time Markov chain
with states `0' and `1' denoting whether the channel is busy or idle, respectively.
The PU's state changes slot by slot according to transition
probabilities $\alpha_0$ and $\beta_0$ shown in Fig.~1.
Clearly, this model is not able to reflect the state of reactive PUs, for which the state transition depends on the SU's past access decisions.
For example, the reactive PU usually reduces its channel access probability if a collision occurs
and increases such probability when the environment becomes friendly again (no
collision is observed).
To reflect PU's reactive behaviors in practice, we propose an enhanced
channel occupancy model.

\begin{figure}[t]
\centering
\DeclareGraphicsExtensions{.eps,.mps,.pdf,.jpg,.png}
\DeclareGraphicsRule{*}{eps}{*}{}
\includegraphics[angle=0, width=0.25\textwidth]{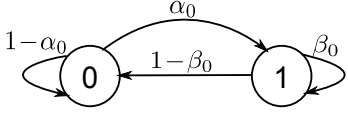}
\caption{Channel occupancy model for the non-reactive PU.}
\label{fig:NR PU}
\end{figure}

\begin{figure}[t]
\centering
\DeclareGraphicsExtensions{.eps,.mps,.pdf,.jpg,.png}
\DeclareGraphicsRule{*}{eps}{*}{}
\includegraphics[angle=0, width=0.4\textwidth]{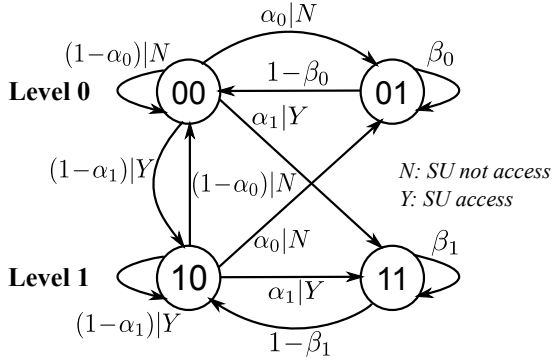}
\caption{Channel occupancy model for the reactive PU.}
\label{fig:RPU}
\end{figure}
The new model is composed of two levels, namely, Level 0 and Level 1.
The reactive PU is assumed to have a higher probability to access the channel when it is in Level 0 than in Level 1.  As a result, the
enhanced channel occupancy model becomes a four-state Markov chain  as shown in Fig.~2, with
each level having two states (busy or idle).
For convenience, we use  2 bits to
represent the total 4 states for the reactive PU. The first bit denotes the level and the
second bit denotes whether the channel is busy (`0') or idle (`1').

In each time-slot, the state of a reactive PU evolves according to the state and the SU's action in the previous slot.
Suppose that initially the PU is in Level 0 with transition probabilities $\alpha_0$ and $\beta_0$, as shown in Fig.~1.
If there is no SU accessing the channel, or the SU accesses when the PU is at state `01', no
collision occurs and the PU will  stay in  Level 0.
However, if the SU accesses  when the PU's state is `00', the PU will react to the resulted collision by transiting from Level 0 to Level 1, with probability $\alpha_1$ to state `11' and probability $1-\alpha_1$ to state `10'.
We assume $\alpha_1 \geq \alpha_0$ to reflect the reduced probability that the reactive PU accesses the channel in Level 1 than in Level 0.
When the PU transits to Level 1, it will stay in this level if the SU continues to access the channel.
However, if the SU does not access the channel and the PU is at state `10', the PU  observes no collision and thus conceives that the environment has become friendly for its transmission.
As a result, the PU increases its probability to access the channel by returning to Level 0, with transition probabilities $\alpha_0$ to state `01' and $1-\alpha_0$ to state `00', respectively.
Since the reactive PU's state transitions are related to the SU's actions at state `00' or `10', the
corresponding transition probabilities are conditioned on the SU's action as shown in Fig.~2.
Moreover, notice that when the PU's state is `01' or `11', the state transition probabilities are not affected by the SU's actions.
This is because no collision occurs if the PU does not attempt to transmit.
We assume $\beta_1 \geq \beta_0$ to be consistent with the reactive PU's more willingness to access in Level 0 than in Level 1.
Note that when $\alpha_1=\alpha_0$ and $\beta_1=\beta_0$, the proposed 4-state channel occupancy model for the reactive PU  reduces to the conventional 2-state  counterpart in Fig.~1 for the non-reactive PU.

It is worth pointing out that our proposed two-level Markov chain model is a basic model that  captures  essential reactions of the PU subject to the SU's collisions; and therefore it can be generalized to specify more complicated reactions of the PU (e.g., random transmission backoff in CSMA)  by appropriately setting the transition probabilities in each level and/or increasing the number of levels in the model.

\subsection{SU's OSA} \label{sec:Net_Model} \label{subsection: SU_OSA}
We assume that the SU can only select one channel for sensing in each time-slot due to hardware limitations, and the sensing result over the selected channel may not be the PU's actual state due to sensing errors.
Similar to \cite{Y.ChenIT08}, the SU makes a sequence of decisions in each slot as follows. At the beginning of slot $t$, $t\geq1$, the SU transmitter selects a channel $a(t) \in \mathbb{A}_{S}$ to sense, where
$\mathbb{A}_{S}=\{1,2,\dots,N\}$ denotes the set of channels.
Supposing  $a(t)=a$, the SU then decides the sensor operating point to sense channel $a$, which is determined by the probability of false alarm (PFA) $\epsilon_{a}(t) \in [0,1]$ and
the probability of mis-detection (PM) $\delta_{a}(t) \in [0,1]$.
A feasible operating point must be confined by the optimal receiver operating characteristic (ROC) curve\footnote{Given the maximum allowable PFA $\epsilon$, the smallest  achievable probability of mis-detection, denoted by $\delta^{*}$, can be attained by the optimal Neyman-Pearson detector \cite{SP}. By varying $\epsilon$ over $[0,1]$, the resultant $\delta^{*}$ and $\epsilon$ pairs form the optimal ROC curve.} and the line determined by $1-\delta_a(t)=\epsilon_a(t)$.
The set of all  feasible operating points is denoted by $\mathbb{A}_{\delta}(a(t))$.
Based on the sensing result $\Theta_a(t) \in \{0,1\}$, the SU decides a pair of access probabilities $\big(f_a(0,t),f_a(1,t)\big) \in [0,1]^2$ for this channel, where $f_a(\theta,t)$
is the access probability on channel $a$ in slot $t$ with $\Theta_a(t)=\theta$.
Denoting $\Phi_a(t)\in\{0(\textrm{not~access}),1(\textrm{access})\}$ as the SU's access action on channel $a$ in slot $t$, $f_a(\theta,t)$ is  expressed by the following conditional probability
\begin{equation}
f_{a}(\theta,t)= P\{\Phi_{a}(t)=1 |\Theta_{a}(t)=\theta\}. \label{f_a}
\end{equation}

At the end of slot $t$, the SU transmitter receives error free feedback $K_a(t) \in \{0,1\}$ from the SU receiver, where $K_a(t)=1$ means that the SU's information is transmitted successfully, and $K_a(t)=0$ represents that the SU transmits but the SU receiver fails to receive the transmitted information due to that the PU is busy and hence a collision occurs.
Note that if the SU does not transmit, the SU transmitter will not receive any feedback. For the ease of representation, we assume that this case is also represented by $K_a(t)=0$.
Note that $K_a(t)$ is for  the SU transmitter and receiver to maintain their decision synchronization \cite{Y.ChenIT08}.

\section{OSA Design in Partially Observable Environments under Reactive PU Model}
As described in Section \ref{subsection: SU_OSA},  in each time-slot, the SU  selects one channel for sensing and thus is unable to observe the PUs'  states in the other $N-1$ channels.
Even for the case of $N=1$, the SU may not be able to obtain the PU's actual state  due to sensing errors. This renders the PUs' states are only \emph{partially} observable at the SU over time.
Thus,  we adopt a POMDP model to design the SU's  OSA.
In this section, we describe the POMDP and formulate the  SU's optimal OSA design  as a constrained POMDP problem.

\subsection{POMDP Elements}
A POMDP in general consists of the following elements \cite{POMDP.Survey}: a set of time-slots $\{1,\ldots,T\}$, where $T$ is called the horizon,  and a set of system states (with transition probabilities), actions, observations (with observation probabilities) and rewards, for each of the time-slots. In this subsection, we formulate the POMDP model for the SU's OSA by specifying these elements.

Specifically, we consider a finite-horizon POMDP with $T<\infty$.
Each system state in the POMDP is denoted by an $N$-element vector, with each element representing one PU's state at its assigned channel.
For brevity, we  represent the states in Fig.~2, namely, $00,01,10,11$, using $0,1,2,3$, respectively, and denote $\mathbb{C}_{S}=\{0,1,2,3\}$ as the set of the states.
Since  $|\mathbb{C}_{S}|=4$, there are in total $4^{N}$ POMDP system states.
Denote the action space of the POMDP as $\mathbb{A}~=~\Big\{\Big(a(t), \big(\epsilon_{a}(t), \delta_{a}(t)\big),\big(f_{a}(0,t),f_{a}(1,t)\big)\Big): a(t)\in \mathbb{A}_{S},\big(\epsilon_{a}(t), \delta_{a}(t)\big)\in \mathbb{A}_{\delta}(a(t)), \big(f_{a}(0,t),f_{a}(1,t)\big) \in [0,1]^2\Big\}$.
Let the observation in the POMDP be $K_a(t) \in \{0,1\}$ in slot $t$. Suppose  that $A(t)=A$, where $A=\big(a(t)=a,(\epsilon_{a}(t), \delta_{a}(t)),(f_a(0,t),f_a(1,t))\big)$. The observation probability is then denoted as $U_{A}(k|i)\triangleq P\{K_a(t)=k |i,A\}$ with $k=\{0,1\}$, which represents the conditional probability of observing $K_a(t)=k$ given that the SU's action is $A$ and the PU's state over selected channel $a$ is $i$, $i\in \mathbb{C}_S$.
Let $\emph{I}(a,t)=0$ and $\emph{I}(a,t)=1$ represent channel $a$ being busy and idle in slot $t$, respectively.
Denote $1_{[x]}$ as an indicator function, which equals 1 if $x$ is true and 0 otherwise.
Note that $K_a(t)=\emph{I}(a,t)\Phi_{a}(t)$.
By applying a derivation similar to that in \cite{Y.ChenIT08}, we  obtain the following result:
\begin{equation}
  U_{A}(k|i)=\left\{
   \begin{array}{l}
   1_{[\emph{I}(a,t)=1]}g_a(t),  \textrm{~if~}k=1\\
   1- U_{A}\{k=1|i\},  \textrm{~if~}k=0
   \end{array}
  \right.  \label{kinU}
\end{equation}
where
\begin{align}
g_a(t)
& \triangleq  P\{\Phi_a(t)=1|a(t)=a,\emph{I}(a,t)=1\} \\
&=\epsilon_a(t)\times f_a(0,t)+(1-\epsilon_a(t))\times f_a(1,t)  \label{g_a}
\end{align}
is the SU's conditional access probability on channel $a$, given that the PU is idle on this channel in slot $t$ and the SU selects channel $a$ to sense.
If the SU transmits successfully in slot $t$, i.e., the observation $K_a(t)=1$, it will obtain a unit throughput.
Denote the instantaneous reward of the SU in slot $t$ by $R_S(t)$.
By assuming a unit bandwidth ($B=1$) for each channel, we have
\begin{equation}
R_S(t)=K_a(t)\times B=K_a(t). \label{ImmediateR}
\end{equation}

At the end of each slot, the POMDP system moves to the next state from the current state according to the POMDP state transition probability.
Since the PUs' states over different channels evolve \emph{independently}, the POMDP system state transition probability is obtained as the product of each PU's state transition probability.
We thus focus on the state transition probability for a given channel $n\in \mathbb{A}_{S}$.
Let $i,j\in \mathbb{C}_{S}$ and denote $P_n(i|j,A)$ as the transition probability from state $j$ in slot $t$ to state $i$ in slot $t+1$ over channel $n$ under the SU's action $A=\big(a(t)=a,(\epsilon_{a}(t), \delta_{a}(t)),(f_a(0,t),f_a(1,t))\big)$.
If $n\neq a$, the SU does not select channel $n$ for sensing.
Fig.~3(a) shows the state transition probabilities for this case, where we use superscript $n$ to denote channel $n$. These probabilities are easily obtained from Fig.~2 with the SU's access action being \emph{not access}.
If $n=a$, the SU selects channel $a$ to sense (and probably access).
Fig.~3(b) shows the state transition probabilities for this case.
If the state is `01' or `11', the transition probabilities are independent of the SU's action; otherwise, according to Fig.~2, they are subject to the SU's access action $\Phi_a(t)$.
Note that given the PU's state on channel $a$, $\Phi_a(t)$ is determined in probability by $(\epsilon_{a}(t), \delta_{a}(t))$ and $(f_a(0,t),f_a(1,t))$. Based on Fig.~2, with the fact that
$P_a(i|j=2,A)= P_a(i|j=0,A),\forall i \in \mathbb{C}_{S}$,  we  thus obtain the following transition probabilities:
$$
 \!\!\left\{
  \begin{array}{l}
  \!\!\!\!P_a(i\!=\!0|j\!=\!0,\!A)\!=\!P_a(i\!=\!0|j\!=\!2,\!A)\!=\!(1\!\!-\!\! \mu_a(t))\!\! \times\! \!(1\!\!-\! \!\alpha_0),\\
  \!\!\!\!P_a(i\!=\!1|j\!=\!0,\!A)\!=\!P_a(i\!=\!1|j\!=\!2,\!A)\!=\!(1\!-\! \mu_a(t))\!\! \times \!\! \alpha_0, \\
  \!\!\!\!P_a(i\!=\!2|j\!=\!0,\!A)\!=\!P_a(i\!=\!2|j\!=\!2,\!A)\!=\! \mu_a(t)\! \!\times \!\!(1\!-\!\alpha_1),  \\
  \!\!\!\!P_a(i\!=\!3|j\!=\!0,\!A)\!=\!P_a(i\!=\!3|j\!=\!2,\!A)\!=\! \mu_a(t)\! \!\times \!\! \alpha_1, 
  \end{array}
 \right.
$$
where
\begin{align}
\mu_a(t)& \triangleq  P\{\Phi_a(t)=1|a(t)=a,\emph{I}(a,t)=0\} \\
&=(1-\delta_a(t))\times f_a(0,t)+\delta_a(t)\times f_a(1,t) \label{mu_a}
\end{align}
is the SU's conditional access probability on channel $a$, given that the PU is busy on this channel in slot $t$ and the SU selects channel $a$ to sense.

\begin{figure}[t]
\centering
\DeclareGraphicsExtensions{.eps,.mps,.pdf,.jpg,.png}
\DeclareGraphicsRule{*}{eps}{*}{}
\includegraphics[angle=0, width=0.36\textwidth]{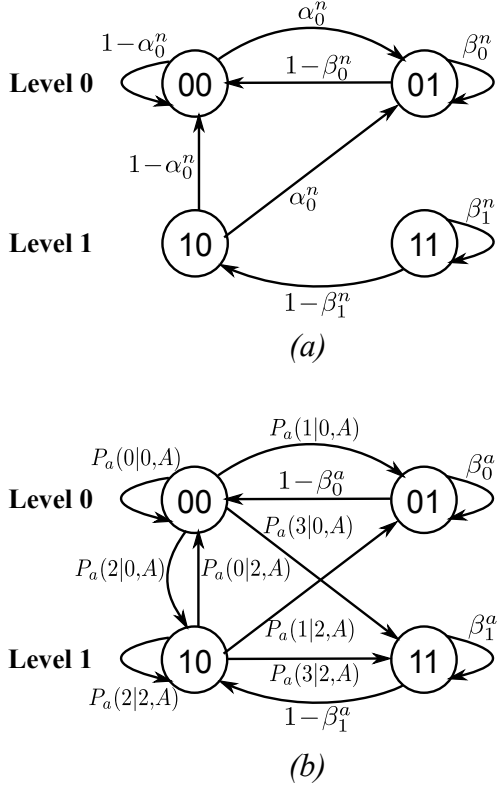}
\caption{State transition probabilities in the reactive PU model: $(a)$ shows the case when $n\neq a$, and $(b)$ shows the case when $n=a$.}
\label{fig:4phases}
\end{figure}

\subsection{Belief on POMDP States} \label{sec: POMDP elements}
In the POMDP model, the system states are not exactly known at the SU. However, based on the SU's previous actions and observations, a belief on the POMDP system state can be obtained.
The belief is defined as the conditional probability distribution over all possible POMDP system states given the history of the SU's actions and observations. As shown in  \cite{POMDP}, the belief on the POMDP system state is a \emph{sufficient statistic} for the design of optimal actions. For our model, the POMDP system state consists of the PUs'  states over independent channels. The SU's belief on the POMDP system state is thus given by the belief on the PU's 4 possible  states over $N$ channels. Hence, we adopt a $4\times N$ matrix $\mathbf{\Lambda}(t)=\big\{\lambda_{nj}(t)\big\}_{n\in \mathbb{A}_S,j\in \mathbb{C}_S}$ as the \emph{belief state} of the POMDP, where the element $\lambda_{nj}(t)$ represents the conditional probability that the state of channel $n\in \mathbb{A}_S$ is $j\in \mathbb{C}_S$ in slot $t$, given the SU's decision and observation history.
We have $\sum_{j\in \mathbb{C}_S}\lambda_{nj}(t)=1$ for $t\in\{1,\ldots,T\}$. Clearly, the space of the POMDP belief states is $[0,1]^{4\times N}$.
The belief state is updated slot by slot based on the SU's previous actions and observations.
Suppose that channel $a$ is selected in slot $t$.
The belief  on states in slot $t+1$ is updated as follows.
\begin{itemize}
\item If $n\neq a$, the updated belief on the state is not affected by the SU's action. We thus have
        \begin {equation}
          \lambda_{nj}(t+1)=\sum_{i\in \mathbb{C}_S} \lambda_{ni}(t)P_n(j|i,A),~\forall j\in \mathbb{C}_S.\label{updatedBM2}~~~~~~~
       \end{equation}
\item If $n=a$, the updated belief on the state is related to the SU's action and is
      obtained according to the observation $K_a(t)=k$ via Bayes rule \cite{POMDP} as
        \begin {equation}
        \! \lambda_{aj}(t+1)\!=\!\frac{\sum_{i\in \mathbb{C}_S}\! \lambda_{ai}(t)P_a(j|i,\!A)U_{\!A}(k|i)}
                           {\sum_{i=0}^{3} \lambda_{ai}(t) U_{\!A}(k|i)},~\forall j\!\in\! \mathbb{C}_S. \label{updatedBM1}
       \end{equation}
\end{itemize}

\subsection{Policy Description}
The OSA design for the SU is  given by a sensing policy $\pi_{s}$,  a sensor operating policy $\pi_{\delta}$ and an access policy $\pi_{c}$. Specifically, the sensing policy specifies a sequence of functions as $\pi_{s}\triangleq \{d_1^{s},d_2^{s},...,d_T^{s}\}$, where $d_t^{s}$ in slot $t$ maps a belief state $\mathbf{\Lambda}(t)$ to the channel $a(t)\in \mathbb{A}_{s}$ selected to sense for this slot.
Given the selected channel $a(t)=a$,  the sensor operating policy specifies a sequence of functions as $\pi_{\delta}\triangleq \{d_1^{\delta},d_2^{\delta},...,d_T^{\delta}\}$, where  $d_t^{\delta}$ in slot $t$ maps  $\mathbf{\Lambda}(t)$ to a feasible sensor operating point $(\epsilon_{a}(t), \delta_{a}(t)) \in\mathbb{A}_{\delta}(a(t))$ for this slot.
Similarly,  the access policy is specified as $\pi_{c}\triangleq \{d_1^{c},d_2^{c},...,d_T^{c}\}$. Given the sensing result $\Theta(t)\in \{0,1\}$ in slot $t$, $d_t^{c}$ maps  $\mathbf{\Lambda}(t)$ to the access probabilities $(f_a(0,t),f_a(1,t))\in [0,1]^2$ in slot $t$.

\subsection{Constrained POMDP Problem}
The optimal OSA design $\{\pi_{s}^{*},\pi_{\delta}^{*}, \pi_{c}^{*}\}$  is obtained by solving a constrained POMDP problem, which maximizes the SU's expected reward over $T$ slots subject to various  constraints to protect the PUs' transmission.
Specifically, the objective of the POMDP problem is to obtain
\begin{eqnarray}
\{\pi_{s}^{*},\!\pi_{\delta}^{*},\!\pi_{c}^{*}\}\!=\!\mathrm{arg}\!
\mathop{\max}_{\pi_{s},\! \pi_{\delta},\! \pi_{c}} E_{\{\pi_{s},\! \pi_{\delta},\! \pi_{c}\}} \Big\{\!\!\sum_{t=1}^{T}R_S(t)|\mathbf{\Lambda}(1)\Big\}, \label{objective}
\end{eqnarray}
where $\mathbf{\Lambda}(1)$ is the initial belief state in slot $t=1$. The elements in $\mathbf{\Lambda}(1)$  are set according to the stationary distribution of the underlying Markov chain, under the assumption that the PU is  at Level 0 initially, i.e., $\lambda_{n0}(1)=\frac{1-\beta_0^n}{1+\alpha_0^n-\beta_0^n}$, $\lambda_{n1}(1)=1-\lambda_{n0}(1)$, and $\lambda_{n2}(1)=\lambda_{n3}(1)=0$, $\forall n\in \mathbb{A}_S$. Suppose that the SU selects channel $a$ in slot $t$. Given  $\mathbf{\Lambda}(t)$, from (\ref{ImmediateR}), the SU's expected reward in slot $t$ over all possible PU's states on channel $a$ and the SU's observations is  obtained as
\begin{equation}
E_{\{\pi_{s}, \pi_{\delta}, \pi_{c}\}} \Big\{R_S(t)|\mathbf{\Lambda}(t)\Big\}=(\lambda_{a1}+\lambda_{a3})\times g_a(t). \label{SU_expect_R}
\end{equation}

We consider two types of protection methods for the reactive PUs, namely, the short-term conditional collision probability (SCCP) constraint and the long-term PU throughput (LPUT) constraint, for which the detailed formulations will be given in Section \ref{section: SCCP} and Section \ref{section: LPUT}, respectively.

\begin{table*}[t]
\caption[The non-monotonicity  of $Q_{t}(\mathbf{\Lambda}(t)|A)$ with respect to $g_a(t)$ under the Reactive-PU model]%
{The non-monotonicity of $Q_{t}(\mathbf{\Lambda}(t)|A)$ with respect to $g_a(t)$ under the Reactive PU model \\ {\textnormal{(with  $\alpha_0^a=0.5,\beta_0^a=0.5,\alpha_1^a=0.9, \beta_1^a=0.9$).}} }
\begin{center} {
\begin{tabular}{|c|c|c|c|c||c|c|}
%\begin {tabular}{|l| >{$}c<{$}|>{$}c<{$}|>{$}c<{$}|>{$}c <{$}|>{$}c<{$}|>{$}cc<{$}|}
\thickhline
\backslashbox { Cases }{ Actions } & $f_a(0,1)$ & $f_a(1,1)$ & $\epsilon_a(1)$ & $\delta_a(1)$ & $g_a(1)$ & $Q_{1}(\mathbf{\Lambda}(1)|A)$ \\
\hline
\textit{case 1}  &     0 & 0.5  &   0.5 & 0.5  &  0.25 & 0.675  \\ \hline
\textit{case 2}  &     0 & 0.6  &   0.5 & 0.5  &  0.3& 0.71  \\  \hline
\textit{case 3}  &     0 & 0.6  &   0.5 & 0.1  &  0.3 & 0.662  \\
\thickhline
\end{tabular} }
\end{center}
\label{tab}
\end{table*}

\section{OSA Design under SCCP Constraint} \label{section: SCCP}
The SCCP constraint has been widely adopted in the literature, e.g., \cite{Y.ChenIT08}-\cite{Tong.JSAC.2011}, to protect the PU's transmission by imposing a conditional collision probability constraint $\zeta$ on  channel $n\in \mathbb{A}_S$, and is defined as
\begin{equation}
\sigma_n(t)\!\triangleq\! P\{\Phi_n(t)\!\!=\! 1|\emph{I}(n,\!t)\!=\!0\}\!\leq\! \zeta,  \forall n\!\in \! \mathbb{A}_S,  \forall t \! \in \! \{1,\ldots,\!T\}. \label{DefinePa}
\end{equation}
The SCCP constraint ensures that on every channel, the PU experiences collisions from the SU for no more than $\zeta$ fraction of the transmission time.
Thus, the PU's throughput under the SU's OSA is at least $100 \times (1-\zeta)$ percentage of that without presence of the SU, if the PU is non-reactive.
However, the  effectiveness of the SCCP constraint  in protecting the reactive PU's transmission remains unaddressed yet in the literature.
In this section, we  adopt the conventional SCCP constraint and design the optimal OSA policy for the SU under the reactive PU model. By adopting the optimal OSA policy, we  show that  the SCCP constraint is not able to provide effective protection to the reactive PU's transmission.

Suppose $a(t)=a$.
With a derivation similar to that in \cite{Y.ChenIT08}, we obtain from (\ref{DefinePa}) that
\begin{equation}
  \!\!\sigma_{n}(t)\!=\!\left\{
   \begin{array}{l}
   \!\!\!(1\!-\!\delta_{a}(t))\!\times\! f_{a}(0,t)\!+\!\delta_{a}(t)\!\times\! f_{a}(1,t),~  \textrm{if}~n=a.\\
   \!\!\!0,~\textrm{if}~n\neq a.  \\
   \end{array}
  \right. \label{colconstraint}
\end{equation}
Note that if $n=a$, $\sigma_a(t)$ has the same expression as $\mu_a(t)$ in (\ref{mu_a}).
Using (\ref{objective}) and (\ref{colconstraint}), the OSA design for the SU under the SCCP constraint is formulated as
\begin{align}
\mathrm{(\textrm{P1})}:~\mathop{\mathrm{max.}}_{\pi_{s}, \pi_{\delta}, \pi_{c}}&~~
E_{\{\pi_{s}, \pi_{\delta}, \pi_{c}\}}\big\{ \sum_{t=1}^{T}R_S(t)|\mathbf{\Lambda}(1)\big\}  \nonumber \\
\mathrm{s.t.} & ~~\sigma_n(t)\leq \zeta,~~~\forall n\in \mathbb{A}_S, ~\forall t\{1,\ldots,T\}. \nonumber
\end{align}

According to the principle of dynamic programming \cite{DP},  (P1) can be decoupled into $T$ subproblems without loss of optimality.
Each subproblem is to find a value function $V_{t}(\mathbf{\Lambda}(t))$, $1 \leq t \leq T$, which represents the SU's maximum expected reward that can be obtained from slot $t$  to slot $T$ under the SCCP constraint, given the belief state $\mathbf{\Lambda}(t)$.
Given the SU's action $A(t)=A$ and the belief state $\mathbf{\Lambda}(t)$, the reward in slot $t$, $1\leq t\leq T-1$, consists of two parts:
the SU's expected immediate reward $E_A\{R_S(t)|\mathbf{\Lambda}(t)\}$ and the SU's maximum expected future reward $E_A\{V_{t+1}(\mathbf{\Lambda}(t+1))|\mathbf{\Lambda}(t)\}$ over all possible  PU's states on channel $a$ and the SU's observations in slot $t$, where $\mathbf{\Lambda}(t+1)$ is updated from $\mathbf{\Lambda}(t)$ according to  (\ref{updatedBM2}) and (\ref{updatedBM1}).
In the last slot $t=T$, the  SU's expected immediate reward alone is the value function $V_{T}(\mathbf{\Lambda}(T))$.
By averaging over all possible PUs' states and the SU's observations $K_a(t)=k \in \{0,1\}$ and maximizing over all SU's actions $A \in \mathbb{A}$, we obtain the value functions expressed as
\begin{align}
&\!\!\!V_{t}(\mathbf{\Lambda}(t))\!=\! \max_{A\in \mathbb{A}}\!\sum_{k=0}^{1}\!\sum_{i=0}^3\!\lambda_{ai}(t)U_{\!A}(k|i)[k\!+\!V_{t\!+\!1}(\mathbf{\mathbf{\Lambda}}(t\!+\!1))], \nonumber \\
&~~~~~~~~~~~~~~~~~~~~~~~~~~~~~~~~~~~~~~~~~~~~~~~~1\leq t\!\leq \!T\!-\!1. \label{ValueF} \\
&\!\!\!V_{T}(\mathbf{\Lambda}(T))\!=\!\max_{A\in \mathbb{A}}\! \sum_{k=0}^{1} \!\sum_{i=0}^3\! \lambda_{ai}(T) U_{\!A}(k|i)\!\times\! k,~~t\!=\!T. \label{VFT}
\end{align}

By computing the value functions given in (\ref{ValueF}) and (\ref{VFT}) recursively backward in time and searching over all possible actions $A(t)\in \mathbb{A}$ in each slot, we can find the optimal policy $\{\pi_{s}^{*},\pi_{\delta}^{*},\pi_{c}^{*}\}$ for (P1) that maximizes the SU's expected reward over $T$ slots, i.e., $V_{1}(\mathbf{\Lambda}(1))$ in (\ref{ValueF}),
under the SCCP constraint.
However, (\ref{ValueF}) and (\ref{VFT}) are generally intractable and computationally prohibitive
due to the infinite and uncountable action space $\mathbb{A}$ \cite{Y.ChenIT08}.

\subsection{Optimal OSA Policy Based on Separation Principle}
In \cite{Y.ChenIT08}, a \emph{separation principle} was proposed to obtain the optimal policy for an OSA design problem similar to (P1), but under the non-reactive PU model.
It is shown in \cite{Y.ChenIT08} that with this method, the sensing policy can be separately designed from the sensor operating policy and the access policy.
The optimal sensor operating policy $\pi_{\delta}^{*}$  and the optimal access policy $\pi_{c}^{*}$ over any selected channel $a$ for sensing is obtained by maximizing  $g_a(t)$ given in (\ref{g_a}), subject to the SCCP constraint in slot $t$.
Since the action space of $\mathbb{A}_S$ is finite and countable, the optimal sensing policy $\pi_{s}^{*}$ is then  obtained by standard dynamic programming techniques, given $\pi_{\delta}^{*}$  and $\pi_{c}^{*}$.

For the ease of presentation, we define  $Q_{t}(\mathbf{\Lambda}(t)|A)$ as the SU's maximum expected reward that can be obtained from slot $t$ to slot $T$, given the SU's action $A\in\mathbb{A}$ in slot $t$ and the belief state $\mathbf{\Lambda}(t)$, i.e.,
\begin{align}
&\!\!\!Q_{t}(\mathbf{\Lambda}(t)|A)\!=\!\sum_{i=0}^3\! \sum_{k=0}^{1}\!\lambda_{ai}(t)U_{\!A}(k|i)[k\!+\!V_{t\!+\!1}(\mathbf{\Lambda}(t\!+\!1))], \nonumber \\
&~~~~~~~~~~~~~~~~~~~~~~~~~~~~~~~~~~~~~~~~~~~~~~~~1\!\leq \!t\!\leq\! T\!-\!1.  \label{Q_t}  \\
&\!\!\!Q_{T}(\mathbf{\Lambda}(T)|A)\!=\!\sum_{i=0}^3\! \sum_{k=0}^{1}\!\lambda_{ai}(t)U_{\!A}(k|i)\!\times\! k, ~~t\!= \!T. \label{Q_T}
\end{align}
Then we have
\begin{equation}
V_t(\mathbf{\Lambda}(t))=\mathrm{arg}\mathop{\max}_{A\in \mathbb{A}} Q_{t}(\mathbf{\Lambda}(t)|A), ~~1\leq t\leq T. \label{V_Q_t}
\end{equation}

As shown in \cite{Y.ChenIT08}, the main reason why the separation principle holds under the non-reactive PU model is that over any selected channel $a$ at slot $t$, $Q_{t}(\mathbf{\Lambda}(t)|A)$ strictly increases with $g_a(t)$ given in (\ref{g_a}). 
However, under the proposed reactive PU model, since the PU's channel access probabilities in the future slots depend on the SU's current action decision, $Q_{t}(\mathbf{\Lambda}(t)|A)$ is generally not monotonically increasing with $g_a(t)$. A simple example with $T=2$ and $N=1$ is shown in TABLE I to validate this observation.
In TABLE I,  we consider three cases, where the SU has different sensor operating points and spectrum access probabilities over channel $a$ in slot $t=1$. We compute $g_a(1)$ and $Q_{1}(\mathbf{\Lambda}(1)|A)$ based on (\ref{g_a}), (\ref{Q_t}) and (\ref{Q_T}) in slot $t=1$ and compare them under these cases.
It is shown that both $g_a(1)$ and $Q_{1}(\mathbf{\Lambda}(1)|A)$ in \emph{case 2} are larger than that in \emph{case 1}; however, the SU's actions in \emph{case 3} give a larger $g_a(1)$ but a  smaller $Q_{1}(\mathbf{\Lambda}(1)|A)$, compared to \emph{case 1}.
Hence, $Q_{t}(\mathbf{\Lambda}(t)|A)$ does not necessarily increase with $g_a(t)$. As a result, the proof in \cite{Y.ChenIT08} for the separation principle does not apply to our problem  under the reactive PU model.

Interestingly, as shown in the following theorem, a separation principle similar to that in  \cite{Y.ChenIT08} is applicable under the reactive PU model without any loss of optimality.
This is true mainly due to the fact that the SCCP constraint in (P1) only depends on $\pi_{\delta}$ and $\pi_c$.
\begin{theorem} \label{theorem: separation principle}
The SU's optimal OSA policy for (P1) under the reactive PU model and the SCCP constraint is obtained by the following two steps:
\begin{itemize}
\item Step 1: Determine the optimal sensor operating policy $\pi_{\delta}^{*}$ and  the optimal access policy $\pi_{c}^{*}$. Specifically, in slot $t$ , supposing $a(t)=a$,
the optimal policies of $\pi_{\delta}^{*}$ and $\pi_{c}^{*}$ are given by
\begin{equation}
  \left\{
   \begin{array}{l}
  \!\! \delta_a^{*}(t)=\zeta, \\
  \!\!\epsilon_a^{*}(t) ~\textrm{is on the optimal ROC curve}\\
  \!\!~~~~~~~\textrm{corresponding to} ~ \zeta, \\
  \!\! f_a^{*}(0,t) = 0,  \\
  \!\! f_a^{*}(1,t) = 1.  
   \end{array}
  \right. \label{Optimalsolution}
\end{equation}
\item Step 2: Apply the optimal policies $\pi^{*}_{\delta}$ and $\pi^{*}_{c}$ in Step 1 to obtain the optimal sensing policy $\pi^{*}_{s}$ by solving the following unconstrained POMDP:
 \begin{equation}
  \pi^{*}_{s}=\mathrm{arg}\mathop{\max}_{\pi_{s}} E_{\pi_{s}}\Big\{\sum_{t=1}^{T} R_S(t)|\mathbf{\Lambda}(1),\pi^{*}_{\delta},\pi^{*}_{c}\Big\}. \label{pi_s}
 \end{equation}
\end{itemize}
\end{theorem}

\begin{proof}
Please refer to Appendix \ref{appendix:proof 1}.\footnote{Note that by setting $\alpha_1^n=\alpha_0^n$ and $\beta_1^n=\beta_0^n$, $\forall n\in \mathbb{A}_s$, the proposed proof  for the separation principle also holds for the non-reactive PU model.}
\end{proof}
\begin{remark}
Since the optimal action decisions, $\delta_a^{*}(t)$, $\epsilon_a^{*}(t)$, $f_a^{*}(0,t)$, and $f_a^{*}(1,t)$, given in (\ref{Optimalsolution}), are independent of the sensing policy, the optimal sensing policy $\pi^{*}_{s}$ can  be separately designed (as shown in Step 2 of Theorem \ref{theorem: separation principle}).
Since all the optimal actions $\delta_a^{*}(t)$, $\epsilon_a^{*}(t)$, $f_a^{*}(0,t)$, and $f_a^{*}(1,t)$ are time-invariant, the optimal polices $\pi^{*}_{\delta}$ and $\pi^{*}_{c}$ are  independent of belief states.
With $f_a^{*}(0,t) = 0$ and  $f_a^{*}(1,t) = 1$,  it follows that the SU always trusts the spectrum sensing result even though there may exist sensing errors, i.e., the SU accesses channel $a$ in slot $t$ with probability 1 when the sensing result is $\Theta_a(t)=1$, and with probability 0 when the sensing result is $\Theta_a(t)=0$.
\end{remark}

\subsection{SCCP Constraint for Protecting  Reactive PUs}
In this subsection, we  show that the SCCP constraint is not sufficient to guarantee the PU's benchmark throughput under the reactive PU model.

We first derive the PU's benchmark throughput on each channel.
Denote the PU's throughput on channel~$n$  in slot $t$ by $R_{P,n}(t)$. If the PU on channel $n$ accesses the assigned channel in slot $t$, $1\leq t \leq T$, and transmits successfully, it will obtain a  unit throughput. Given the current belief state $\mathbf{\Lambda}(t)$ and the SU's OSA policies  $\pi_{s}$, $\pi_{\delta}$, and $\pi_{c}$, it is then easy to obtain that the PU's expected throughput on channel $n$  over  all possible states in slot $t$ is
\begin{equation}
 E_{\pi_{s}, \pi_{\delta}, \pi_{c}}\big\{R_{P,n}(t)|\mathbf{\Lambda}(t)\big\}=P\{\emph{I}(n,t)=0\}\times (1-\sigma_n(t)), \label{PU_IMMR_exp}
\end{equation}
where  $P\{\emph{I}(n,t)=0\}=\lambda_{n0}(t)+\lambda_{n2}(t)$ is the PU's access probability on channel $n$ in slot $t$, and $\sigma_n(t)$ is  given in (\ref{colconstraint}).
By summing the PU's  throughput over $T$ slots and dividing the sum by $T$,  the PU's normalized throughput on channel $n$, denoted by $R_{P,n}^{o}$, is given as
\begin{equation}
 R_{P,n}^{o}= \frac{1}{T}\times E_{\pi_{s}, \pi_{\delta}, \pi_{c}}\big\{\sum_{t=1}^{T} R_{P,n}(t)|\mathbf{\Lambda}(1)\big\}, ~\forall n\in \mathbb{A}_S.
\label{eq: PU_normal_Throughput}
\end{equation}

Note that under the non-reactive PU model, the PU's channel access probability is independent of the SU's spectrum access policy and thus remains the same in each slot.
Specifically, from the stationary distribution of the underlying Markov chain, we have $P\{\emph{I}(n,t)=0\}=\frac{1-\beta_0^n}{1+\alpha_0^n-\beta_0^n}$, $n\in \mathbb{A}_S$, $t\in\{1,\ldots,T\}$.
Since  $\sigma_n(t)\leq \zeta$, from (\ref{PU_IMMR_exp}), we obtain PU's minimum achievable expected throughput in slot $t$ on channel $n$ as $\frac{1-\beta_0^n}{1+\alpha_0^n-\beta_0^n} \times (1-\zeta)$, where  $1-\zeta$ is the minimum probability that the SU does not collide with the PU in slot $t$ on channel $n$. Then according to (\ref{eq: PU_normal_Throughput}),  we obtain PU's minimum achievable normalized throughput on channel $n$ under the non-reactive PU model as
\begin{equation}
\Upsilon_n = \frac{1-\beta_0^n}{1+\alpha_0^n-\beta_0^n} \times (1-\zeta), ~\forall n\in \mathbb{A}_S. \label{reference_throughput}
\end{equation}
Taking $\Upsilon_n$  as the  benchmark throughput for PU on channel $n$, we say that the  PU system is under \emph{effective protection} if $R_{P,n}^{o}\geq \Upsilon_n$, $\forall n\in \mathbb{A}_S$,  is guaranteed.

Next, we show that under the  SCCP constraint,  the reactive PU's normalized throughput is not  always larger than or equal to  the benchmark throughput.  We consider two cases. One is the single-channel case with $N=1$, and the other is the multi-channel case with $N>1$.

For the case of $N=1$ with $n=a$, since the SU always selects to sense channel $a$ and probably accesses it, $N=1$ can be considered the PU's \emph{worst case} transmission. The following proposition shows that the SCCP constraint is not able to provide effective protection for PU's transmission with $N=1$.
\begin{proposition} \label{proposition: ineffectiveness of SCCP}
With the optimal OSA policy in (\ref{Optimalsolution}), for $N=1$ and $T>1$,  the reactive PU's normalized throughput on channel $a$  is strictly smaller than the benchmark throughput, i.e., $R_{P,a}^{o}< \Upsilon_a$.
\end{proposition}
\begin{proof}
Please refer to Appendix \ref{appendix:proof 2}.
\end{proof}

For the case of $N>1$, since the SU can select one from $N$ channels to sense, $R_{P,n}^{o}$, $\forall n\in \mathbb{A}_n$, will be at least equal to that in the worst case with $N=1$.   However, under the reactive PU model with $T>1$, since $R_{P,a}^{o}< \Upsilon_a$ for $N=1$, it is difficult to analyze whether the SCCP constraint is an effective PU protection method for $N>1$. As will be shown later by simulations in  Section VI, the SCCP constraint is not sufficient to guarantee the benchmark throughput of all the $N$ reactive PUs  when $N>1$. Thus, the SCCP constraint is not able to provide effective protection to the PU transmissions.

\section{OSA Design under the LPUT Constraint} \label{section: LPUT}
The long-term PU throughput (LPUT) constraint has been widely adopted in the SS-based CR systems, to guarantee that the PU's transmission quality is always above a predefined threshold regardless of the PU's on/off status  \cite{Levorato.Allerton.12}-\cite{R.ZhangGlobecom08}. In contrast, in the OSA-based CR systems, due to the simplicity and effectiveness of the SCCP constraint in protecting the non-reactive PU  transmissions, the complicated LPUT constraint has not been used to protect PU transmissions, to our best knowledge.

For the reactive PU  transmissions, however, as shown in Section IV-B, the traditional SCCP constraint cannot be adopted as an effective protection method. In this section, we adopt the LPUT constraint as the protection method, which is formulated as
\begin{equation}
R_{P,n}^{o}\geq \Upsilon_n, ~~\forall n\in \mathbb{A}_{S}, \label{eq: PU_throughput_constraint}
\end{equation}
where  $R_{P,n}^{o}$ is defined in (\ref{eq: PU_normal_Throughput}) and $\Upsilon_n$ is given in (\ref{reference_throughput}).
Clearly, the LPUT constraint formulated in (\ref{eq: PU_throughput_constraint}) is able to provide  effective protection to the reactive PU  transmissions, if  it is satisfied by  the SU's OSA policy.
Different from the SCCP constraint in (\ref{DefinePa}), from  (\ref{eq: PU_normal_Throughput}),   the LPUT constraint takes into account the PU's reaction to the SU's collision  in each slot.

By using (\ref{objective}) and (\ref{eq: PU_throughput_constraint}), the OSA design under the LPUT constraint is formulated as
\begin{align}
\mathrm{(\textrm{P2})}:\mathop{\mathrm{max.}}_{\pi_{s},\pi_{\delta},\pi_{c}} &
~~E_{\pi_{s},\pi_{\delta},\pi_{c}}\big\{ \sum_{t=1}^{T}R_S(t)|\mathbf{\Lambda}(1)\big\}  \nonumber \\
\mathrm{s.t.}  &~~R_{P,n}^{o}\geq \Upsilon_n, ~~\forall n\in \mathbb{A}_{S}.   \nonumber
\end{align}

To our best knowledge, there is no existing work that addresses the {\it finite-horizon} {\it long-term constrained} POMDP problem (P2). As the problem has infinite and unaccountable action space, it is challenging to find the optimal policy $\{\pi_{s}^{*},\pi_{\delta}^{*},\pi_{c}^{*}\}$ for it. 
Note that for (P1), we have proposed a separation principle to design $\pi_{s}^{*}$ separately from $\pi_{\delta}^{*}$ and $\pi_{c}^{*}$, since the SCCP constraint in (P1) is only related to $\pi_{\delta}^{*}$, and $\pi_{c}^{*}$. However,  the LPUT constraint in (P2) is determined by all the three policies. Thus,  $\pi_{s}^{*}$ is generally \emph{not independent} of $\pi_{\delta}^{*}$ and $\pi_{c}^{*}$, which implies that (P2) cannot be solved optimally by the separation principle. However, a suboptimal policy for (P2) can be found based on the separation principle. In this section, we first focus on the single-channel case with $N=1$  for (P2), where only $\pi_{\delta}^{*}$ and $\pi_{c}^{*}$ need to be determined. We then propose the suboptimal policy for the general multi-channel case with $N>1$ by extending the results from $N=1$ to $N>1$ based on the separation principle.

\subsection{Single-Channel Case: Optimal OSA Policy Structure}
For (P2) in the single-channel case with $n=a$, we have
\begin{align}
\mathrm{(\textrm{P2-S})}:\mathop{\mathrm{max.}}_{\pi_{\delta},\pi_{c}} &
~~E_{\pi_{\delta},\pi_{c}}\big\{ \sum_{t=1}^{T}R_S(t)|\mathbf{\Lambda}(1)\big\}  \nonumber \\
\mathrm{s.t.}  &~~R_{P,a}^{o}\geq \Upsilon_a,   \nonumber
\end{align}
with action space
\begin{align*}
&\Big\{\big( (\epsilon_{a}(t),\! \delta_{a}(t) ),  (f_{\!a}(0,\!t),\!f_{\!a}(1,\!t) )\big):\nonumber \\
&~~~~~~~~~~~\big(\epsilon_{a}(t),\! \delta_{a}(t)\big)\!\in\!\mathbb{A}_{\delta}(a),\big(f_{\!a}(0,\!t),f_{\!a}(1,\!t)\big)\in[0,\!1]^2\Big\}.
\end{align*}
To simplify (P2-S), in this subsection,  we first propose an equivalent problem to (P2-S), namely, (P2-S-1), through which we  find  the optimal OSA policy structure. Based on the optimal OSA policy structure, we reduce (P2-S-1)  to another  problem (P2-S-2) with a significantly reduced action space.
To facilitate our analysis, we first present the following proposition.
\begin{proposition} \label{proposition: PU_throughput_constraint_with_equality}
For the case of $N\!=\!1$, given any SU's  OSA policy $\pi\!=\!\{\pi_{\delta},\pi_{c}\}$, with the resultant PU's normalized throughput  $R_{P,a}^{o}\!\!>\!\!\Upsilon_a$ and the SU's reward $E_{\pi}\!\Big\{\!\! \sum_{t=1}^{T}\!R_S(t)|\mathbf{\Lambda}(1)\!\Big\}$, we can find another policy $\pi^{'}\!=\!\Big\{\pi_{\delta}^{'},\pi_{c}^{'}\Big\}$,
with the resultant PU's normalized throughput $R_{P,a}^{o~'} = \Upsilon_a$ and  the SU's reward $E_{\pi^{'}}\big\{ \sum_{t=1}^{T}R_S(t)|\mathbf{\Lambda}(1)\big\}> E_{\pi}\big\{ \sum_{t=1}^{T}R_S(t)|\mathbf{\Lambda}(1)\big\}$.
\end{proposition}
\begin{proof}
 Please refer to Appendix \ref{appendix:proof 3}.
\end{proof}

According to Proposition \ref{proposition: PU_throughput_constraint_with_equality}, the optimal policy $\pi^{*}\!=\!\{\pi_{\delta}^{*},\pi_{c}^{*}\}$ is selected to  ensure that $R_{P,a}^{o}\! =\! \Upsilon_a$. Thus,  (P2-S) is \emph{equivalent} to
\begin{align}
\mathrm{(\textrm{P2-S-1})}:\mathop{\mathrm{max.}}_{\pi_{\delta},\pi_{c}} &
~~E_{\pi_{\delta},\pi_{c}}\big\{ \sum_{t=1}^{T}R_S(t)|\mathbf{\Lambda}(1)\big\}  \nonumber \\
\mathrm{s.t.}  &~~R_{P,a}^{o}= \Upsilon_a.   \nonumber
\end{align}

\begin{proposition} \label{proposition: optimal structure}
The structure of the optimal policy $\pi^{*}=\{\pi_{\delta}^{*},\pi_{c}^{*}\}$ for (P2-S-1) is given as follows.
\begin{equation}
  \left\{
   \begin{array}{l}
   \!\!\delta_a^{*}(t)  ~\textrm{is in general time-variant and to be determined},\\
  \!\!\epsilon_a^{*}(t) ~\textrm{is on the optimal ROC curve corresponding}~  \\
  \!\!~~~~~~~\textrm{to} ~ \delta_a^{*}(t), \\
   \!\!f_a^{*}(0,t) = 0,  \\
   \!\!f_a^{*}(1,t) = 1.  
   \end{array}
  \right. \label{OptimalsolutionStructure}
\end{equation}
\end{proposition}
\begin{proof}
Please refer to Appendix \ref{appendix:proof 4}.
\end{proof}

\begin{remark}
As shown in (\ref{OptimalsolutionStructure}), the optimal spectrum access policy $\pi_{c}^{*}$ for (P2-S-1) is the same as that for (P1), and the optimal sensor operating point $\big(\delta_a^{*}(t),\epsilon_a^{*}(t)\big)$ for (P2-S-1) is on the optimal ROC curve as that for (P1). However, different from the time-invariant case for (P1), the optimal PM decision for (P2-S-1) is in general time-varying. As proved in Appendix \ref{appendix:proof 4}, $\delta_a^{*}(t)$ for (P2-S-1) is related to the current belief state $\mathbf{\Lambda}(t)$ and thus needs to be determined adaptively over time.
This indicates  that the spectrum sensor design plays a crucial role in protecting  reactive PU's transmission  under  LPUT constraint.
\end{remark}

By applying (\ref{eq: PU_normal_Throughput}) and (\ref{OptimalsolutionStructure}) and without loss of optimality, (P2-S-1) is  reduced to
\begin{align}
\mathrm{(\textrm{P2-S-2})}:\mathop{\mathrm{max.}}_{\pi_{\delta}} &
~~E_{\pi_{\delta}}\big\{ \sum_{t=1}^{T}R_S(t)|\mathbf{\Lambda}(1),\pi_c^{*}\big\}  \nonumber \\
\mathrm{s.t.}  &~~\frac{1}{T}E_{\pi_{\delta}}\big\{ \sum_{t=1}^{T}R_{P,a}(t)|\mathbf{\Lambda}(1),\pi_c^{*}\big\}= \Upsilon_a,   \nonumber
\end{align}
where $\big(\delta_a(t),\epsilon_a(t)\big)$ determined by $\pi_{\delta}$ is on the optimal ROC curve. Thus, to find the optimal policy $\pi_{\delta}^{*}$ for (P2-S-2), we only need to search the action space of $\big\{\delta_a(t):\delta_a(t)\!\in \![0,1]\big\}$, which is  greatly reduced over that of (P2-S-1).

Since (P2-S-2) is reduced from (P2-S-1) and (P2-S-1) is equivalent to (P2-S), substituting the optimal $\pi_{\delta}^{*}$ for (P2-S-2) into (\ref{OptimalsolutionStructure}) yields the optimal OSA policy for (P2-S-1) and thus (P2-S). Hence, in the following, we focus on solving (P2-S-2).

However, finding $\pi_{\delta}^{*}$  for (P2-S-2) is of high complexity, mainly due to the following two reasons: 1) the infinite and unaccountable action space of (P2-S-2), and 2) the non-deterministic POMDP belief state transitions. As the complexity due to the first reason is obvious, here we explain the complexity due to the second reason. From (\ref{ValueF}) and (\ref{VFT}), to maximize the SU's throughput in (P2-S-2), i.e., to find $V_1(\mathbf{\Lambda}(1))$ under the LPUT constraint,
we need to obtain $V_t(\mathbf{\Lambda}(t))$  for all $t\in \{2,\ldots,T\}$. As shown in (\ref{updatedBM1}), given $\mathbf{\Lambda}(t)$   and the SU's OSA actions  in slot $t$,  2 possible belief states in slot $t+1$ exist with non-zero probability, corresponding to the 2 possible observations, respectively. Thus, for the case of $N=1$, given the initial belief state $\mathbf{\Lambda}(1)$ and the SU's OSA policies $\pi_{\delta}$ and $\pi_c^{*}$, the  complexity of computing $V_1(\mathbf{\Lambda}(1))$ is $\mathcal{O}\big(2^T\big)$, which is not scalable with $T$.

In the following, we focus on designing a suboptimal policy for (P2-S-2) that can meet the LPUT constraint.
We use the method given in Appendix B to calculate the PU's normalized throughput, which is similar to  the SU's throughput calculation in (\ref{Q_t}) and (\ref{Q_T}). However, also due to the non-deterministic POMDP belief state transitions,  it is of exponentially increased complexity over time to find a policy that can meet the LPUT constraint. Motivated by \cite{POMDP.Survey}, we note that the complexity  due to the non-deterministic POMDP belief state transitions is  reducible,  by converting the POMDP into an equivalent MDP with deterministic state transitions. In the following subsections, we first construct the equivalent MDP, and then based on the MDP, we propose a suboptimal policy for (P2-S-2), which satisfies the LPUT constraint in (P2-S-2).

\subsection{Equivalent MDP with Deterministic State Transitions}

In this subsection, we  first convert the POMDP for (P2-S-2) into an  MDP with deterministic state transitions,  and reformulate  the LPUT constraint in  (P2-S-2) based on the MDP. We then show that  if a   policy   satisfies the MDP-based LPUT constraint, it will also satisfy the POMDP-based counterpart in (P2-S-2).

An MDP in general consists of the following elements \cite{MDP1994}: a set of time-slots $\{1,\ldots,T\}$,  a set of system states (with transition probabilities), actions and rewards, for each of the time-slots. In the following, we formulate the MDP model for the SU's OSA by specifying these elements  according to \cite{POMDP.Survey}.
Specifically, for the single-channel case with $n=a$,  the MDP state in slot $t$, $1\leq t \leq T$, is denoted by a 4-element vector $\mathbf{\Omega}(t)=\big\{\omega_{ai}(t)\big\}_{i\in \mathbb{C}_{S}}$, where $\omega_{ai}(t)\in[0,1]$ is the conditional probability that the reactive PU is at the $i$-th state  on channel $a$ in slot $t$,  given the SU's action history. Note that the MDP state space, given by $[0,1]^4$,  is the same as the POMDP belief state space, given in Section III-B, for $N=1$.
We assume that $\mathbf{\Omega}(1)=\mathbf{\Lambda}(1)$ in the initial slot $t=1$.
Based on the current MDP state $\mathbf{\Lambda}(t)$ in slot $t$,  the SU selects an action $A(t)=\big((\epsilon_{a,\mathcal{M}}(t), \delta_{a,\mathcal{M}}(t)), (f_{a,\mathcal{M}}(0,t),f_{a,\mathcal{M}}(1,t))\big)$ for OSA, where  $(\epsilon_{a,\mathcal{M}}(t), \delta_{a,\mathcal{M}}(t))\in \mathbb{A}_{\delta}(a)$ is the sensor operating point and $(f_{a,\mathcal{M}}(0,t),f_{a,\mathcal{M}}(1,t))\in [0,1]^2$ is the channel access probability. Thus, the MDP action space  is  the same as  the POMDP action space for $N=1$. We  then follow the  optimal OSA policy structure in (\ref{OptimalsolutionStructure}) and set
\begin{align}
  &\left\{
   \begin{array}{l}
   \!\!(\epsilon_{a,\mathcal{M}}(t), \delta_{a,\mathcal{M}}(t)) ~\textrm{locates on the optimal ROC curve}, \\
   \!\!f_{a,\mathcal{M}}(0,t)=0,  \\
   \!\!f_{a,\mathcal{M}}(1,t)=1,  
   \end{array}
  \right. \nonumber \\
  &~~~~\forall t\in \{1,\ldots,T\}, \label{MDP_action}
\end{align}
where in slot $t$, the SU only needs to determine the PM action $\delta_{a,\mathcal{M}}(t)$.
Denote $P_a\big(\mathbf{\Omega}(t+1)|\mathbf{\Omega}(t),\delta_{a,\mathcal{M}}(t)\big)$ as the MDP state transition probability from state $\mathbf{\Omega}(t)\!=\!\big\{\omega_{ai}(t)\big\}_{i\in \mathbb{C}_{S}}$ in slot $t$ to state $\mathbf{\Omega}(t+1)\!=\!\big\{\omega_{ai}(t+1)\big\}_{i\in \mathbb{C}_{S}}$ in slot $t+1$ on channel $a$,  given the SU's selected PM action $\delta_{a,\mathcal{M}}(t)$ in slot $t$. From Fig.~3(b),  $\omega_{ai}(t+1)\!=\!\sum_{j=0}^{3}\omega_{aj}(t)P_a\big(i|j,\delta_{a,\mathcal{M}}(t)\big)$,  $i$, $j\!\in\! \mathbb{C}_{S}$,  where $P_a\big(i|j,\delta_{a,\mathcal{M}}(t)\big)$ is the state transition probability $P_a\big(i|j,A\big)$ given in Fig.~3(b), with  $A$ reduced to $\delta_{a,\mathcal{M}}(t)$ by applying (\ref{MDP_action}). Thus, we obtain the MDP state transition probability as
\begin{align}
   &P_a\big(\mathbf{\Omega}(t+1)|\mathbf{\Omega}(t),\delta_{a,\mathcal{M}}(t)\big)\nonumber \\
   &\!=\left\{
   \begin{array}{l}
   \!\!1,\textrm{if}~ \omega_{ai}(t\!+\!1)\!\!=\!\!\sum_{j=0}^{3}\omega_{aj}(t)P_a\big(i|j,\delta_{a,\!\mathcal{M}}(t)\big), \forall i\!\in \! \mathbb{C}_{S}.  \\
   \!\!0,\textrm{otherwise}.
   \end{array}  \label{MDP_TP}
  \right.
\end{align}
From (\ref{MDP_TP}),  the MDP state transition is  deterministic. That is, given  $\delta_{a,\mathcal{M}}(t)$ selected at MDP state $\mathbf{\Omega}(t)$ in slot $t$, there is only one possible MDP state $\mathbf{\Omega}(t+1)$ in  slot $t+1$.

Denote the PU's  throughput in slot $t$ on channel $a$ by $R_{P,a}^{\mathcal{M}}(t)$, which is
\begin{equation}
 R_{P,a}^{\mathcal{M}}(t)=\big(\omega_{a0}(t)+\omega_{a2}(t)\big)\times (1-\delta_{a,\mathcal{M}}(t)).  \label{PU_MDP_ImmR}
\end{equation}
The PU's normalized throughput on channel $a$ over $T$ slots is thus given by $\frac{1}{T}\sum_{t=1}^{T}R_{P,a}^{\mathcal{M}}(t)$. With the benchmark throughput $\Upsilon_a$, given in (\ref{reference_throughput}), the MDP-based LPUT constraint   is formulated as
\begin{equation}
\frac{1}{T}\sum_{t=1}^{T}R_{P,a}^{\mathcal{M}}(t)=\Upsilon_a. \label{LPUT_constraint_MDP}
\end{equation}

Note that due to the deterministic  MDP state transitions,  with a complexity analysis similar to that in Section~V-A, it is easy to find that the complexity in computing $\sum_{t=1}^{T}R_{P,a}^{\mathcal{M}}(t)/T$ under the MDP policy is $\mathcal{O}(T)$, which is substantially reduced as compared to that based on the POMDP.

\begin{proposition} \label{proposition: optimal transformed MDP}

Given an MDP policy  $\pi_{\delta,\mathcal{M}}^{\star}$, which specifies a PM decision  $\delta_{a,\mathcal{M}}^{\star}(t)$ in slot $t$, $1\leq t \leq T$, we  construct a POMDP policy $\pi_{\delta}^{\star}$ for (P2-S-2), where the corresponding PM decision in each slot $t$  is $\delta_a^{\star}(t)=\delta_{a,\mathcal{M}}^{\star}(t)$. If the MDP-based LPUT constraint in (\ref{LPUT_constraint_MDP}) is satisfied under  $\pi_{\delta,\mathcal{M}}^{\star}$,  the POMDP-based LPUT constraint in (P2-S-2) is also satisfied under   $\pi_{\delta}^{\star}$.
\end{proposition}
\begin{proof}
Please refer to Appendix \ref{appendix:proof 5}.
\end{proof}

\subsection{Suboptimal  Policy}
In this subsection, by studying the MDP-based  LPUT constraint, we  derive a suboptimal policy for (P2-S), such that  the LPUT constraint in (P2-S) is satisfied. In the following, we first give a sufficient condition for satisfying the MDP-based LPUT constraint, based on which, we propose an MDP-based policy $\pi_{\delta,\mathcal{M}}^{\star}$, which can satisfy the MDP-based LPUT constraint. Based on $\pi_{\delta,\mathcal{M}}^{\star}$ and Proposition \ref{proposition: optimal transformed MDP}, we  then obtain a suboptimal policy for (P2-S-2), such that the LPUT constraint in (P2-S-2) is satisfied.  Finally, we obtain the suboptimal policy for (P2-S) by substituting the suboptimal policy for (P2-S-2) into the optimal OSA policy structure in (\ref{OptimalsolutionStructure}).

\subsubsection{A sufficient condition for satisfying MDP-based LPUT constraint }
From  (\ref{LPUT_constraint_MDP}), to satisfy the LPUT constraint, we take   $\Upsilon_a\times T$ as the PU's throughput requirement over all $T$ slots in the MDP.  We then denote $X_{P,a}(t)$ as the PU's throughput requirement from slot $t$ to slot $T$ in the MDP and have
\begin{align}
X_{P,a}(1)&\!\triangleq \! \Upsilon_a \! \times \! T,~t\!=\!1, \label{eq: X_P_1} \\
X_{P,a}(t)&\! \triangleq \! X_{P,a}(t\!-\!1)\!-\!R_{P,a}^{\mathcal{M}}(t\!-\!1), ~\forall t \!\in \! \{2,\ldots,T\}. \label{eq: X_P_t}
\end{align}
Given the PU's obtained throughput $R_{P,a}^{\mathcal{M}}(t-1)$ in slot $t-1$, from (\ref{eq: X_P_t}) and by calculating backward in time, we observe that if $X_{P,a}(t)$ in slot $t$ is satisfied, $X_{P,a}(t-1)$ in slot $t-1$ is achieved.
Thus, we can easily show that if $X_{P,a}(T)$ in the last slot $T$ is satisfied, the PU's throughput requirement $X_{P,a}(1)$ is achieved, i.e.,  the LPUT constraint given in (\ref{LPUT_constraint_MDP})  is met.
Note that $X_{P,a}(T)$ can be satisfied by selecting    $\delta_{a,\mathcal{M}}(T)$  such that $X_{P,a}(T)=(\omega_{a0}(T)+\omega_{a2}(T))\times (1-\delta_{a,\mathcal{M}}(T))$.
Since  $\delta_{a,\mathcal{M}}(T)\in [0,1]$,  the following  inequality is obtained as a \emph{sufficient condition} for satisfying (\ref{LPUT_constraint_MDP}):
\begin{equation}
 0\leq X_{P,a}(T)\leq \omega_{a0}(T)+\omega_{a2}(T). \label{X_T}
\end{equation}

\subsubsection{MDP-based policy $\pi_{\delta,\mathcal{M}}^{\star}$}
The MDP-based policy $\pi_{\delta,\mathcal{M}}^{\star}$ is given by the  PM actions $\{\delta_{a,\mathcal{M}}(1),..., \delta_{a,\mathcal{M}}(T)\}$. In the following, we derive the minimum required PM, denoted by $\delta_{a,\mathcal{M}}^L(t)$, and the maximum allowable PM, denoted by $\delta_{a,\mathcal{M}}^U(t)$, in slot $t$, such that  (\ref{X_T}) is satisfied if the SU selects $\delta_{a,\mathcal{M}}(t)\in [\delta_{a,\mathcal{M}}^L(t),\delta_{a,\mathcal{M}}^U(t)]$,  $\forall t\in\{1,\ldots,T\}$.

From (\ref{eq: X_P_1}) and (\ref{eq: X_P_t}), to ensure $X_{P,a}(T)\!\geq \! 0$,  we need $X_{P,a}(t)\! \geq \! 0$ in  all the previous  slots with  $ t \!\in \! \{1,\ldots,T\!-\!1\}$.
By substituting $X_{P,a}(t)\! \geq \! 0$ to (\ref{eq: X_P_t}) and using (\ref{PU_MDP_ImmR}), we obtain $\delta_{a,\mathcal{M}}(t) \! \geq \! 1\!-\!\frac{X_{P,a}(t)}{\omega_{a0}(t)\!+\!\omega_{a2}(t)}$, $1\!\leq \! t \!\leq \!T\!-\!1$.
When $t\!=\!T$, if $X_{P,a}(T)\leq \omega_{a0}(T)\!+\!\omega_{a2}(T)$ is satisfied, we only need to set $\delta_{a,\mathcal{M}}(T)\!=\! 1\!-\!\frac{X_{P,a}(T)}{\omega_{a0}(T)\!+\!\omega_{a2}(T)}$ to satisfy  $X_{P,a}(T)$. Thus, we obtain
\begin{equation}
\delta_{a,\mathcal{M}}^{L}(t)= \max\Big(0,1-\frac{X_{P,a}(t)}{\omega_{a0}(t)+\omega_{a2}(t)}\Big),~1\leq t \leq T.
\end{equation}

\begin{proposition} \label{proposition: Upper bound}
To ensure  $X_{P,a}(T) \leq \omega_{a0}(T)+\omega_{a2}(T)$,  the SU's PM  $\delta_{a,\mathcal{M}}(t)$ selected in slot $t$ needs to satisfy the following inequality:
\begin{align}
\!\!\!\!\!\!\delta_{a,\mathcal{M}}(t) \leq& \frac{\omega_{a1}(t)\times m_2(t)+\omega_{a3}(t)\times m_3(t)-X_{P,a}(t)}{(\omega_{a0}(t)+\omega_{a2}(t))\times m_4(t)} \nonumber \\
& +\frac{m_1(t)}{m_4(t)}\label{upper}
\end{align}
where
\begin{equation}
  \left\{
   \begin{array}{l}
   \!\!m_1(t)\!=\!1\!+\!(1\!-\!\alpha_0^a)\!\times\! m_1(t\!+\!1)\!+\!\alpha_0^a\!\times\! m_2(t\!+\!1),  \\
   \!\!m_2(t)\!=\!(1\!-\!\beta_0^a)\!\times\! m_1(t\!+\!1)\!+\!\beta_0^a\!\times\! m_2(t\!+\!1),  \\
   \!\!m_3(t)\!=\!(1\!-\!\beta_1^a)\!\times\! m_1(t\!+\!1)\!+\!\beta_1^a\!\times\! m_3(t\!+\!1), \\
   \!\!m_4(t)\!=\!1\!+\!(\alpha_1^a\!-\!\alpha_0^a)\!\times\! m_1(t\!+\!1)\!+\!\alpha_0^a\!\times \!m_2(t\!+\!1)\!\\
   \!\!~~~~~~~~~~-\!\alpha_1^a\!\times\! m_3(t\!+\!1),
   \end{array}
  \right.   \label{para_not_last_T}
\end{equation}
for $1\leq t<T$ and
\begin{equation}
  \left\{
   \begin{array}{l}
   m_1(T)=1,  \\
   m_2(T)=0,  \\
   m_3(T)=0, \\
   m_4(T)=1,
   \end{array}
  \right.  \label{para_last_T}
\end{equation}
for $t=T$.
\end{proposition}
\begin{proof}
Please refer to  Appendix~\ref{appendix:proof 6}.
\end{proof}

With Proposition \ref{proposition: Upper bound}, the SU's maximum allowable PM $\delta_{a,\mathcal{M}}^{U}(t)$ is obtained as
\begin{align}
&\!\!\delta_{a,\mathcal{M}}^{U}(t)=\min\bigg(1,\nonumber \\
&\frac{\omega_{a1}(t)\times m_2(t)+\omega_{a3}(t)\times m_3(t)-X_{P,a}(t)}{(\omega_{a0}(t)+\omega_{a2}(t))\times m_4(t)}+\frac{m_1(t)}{m_4(t)}\bigg).
\end{align}

Clearly, if  $\delta_{a,\mathcal{M}}(t) \in [\delta_{a,\mathcal{M}}^L(t),\delta_{a,\mathcal{M}}^U(t)]$  in slot $t$ is selected, (\ref{X_T}) is guaranteed and thus the MDP-based LPUT constraint given in (\ref{LPUT_constraint_MDP}) is satisfied.
We then propose the MDP-based policy $\pi_{\delta,\mathcal{M}}^{\star}$ by specifying
\begin{equation}
\delta_{a,\mathcal{M}}^{\star}(t)=\delta_{a,\mathcal{M}}^{L}(t)+\psi(t)\times (\delta_{a,\mathcal{M}}^{U}(t)-\delta_{a,\mathcal{M}}^{L}(t)),
\end{equation}
where $\delta_{a,\mathcal{M}}^{\star}(t)\in [\delta_{a,\mathcal{M}}^L(t),\delta_{a,\mathcal{M}}^U(t)]$ is guaranteed by selecting $\psi(t)\in [0,1]$ in slot $t$.

\subsubsection{POMDP-based suboptimal policy}

We first consider the POMDP-based problem (P2-S-2). From Proposition \ref{proposition: optimal transformed MDP}, by setting $\delta_a^{\star}(t)=\delta_{a,\mathcal{M}}^{\star}(t)$, $1\leq t \leq T$, we find a suboptimal policy $\pi_{\delta}^{\star}$ for (P2-S-2),  which  satisfies the LPUT constraint in  (P2-S-2).

Next, we consider the original POMDP-based problem (P2-S). Note that (P2-S-2) is reduced from (P2-S-1) without loss of optimality and  (P2-S-1) is equivalent to (P2-S). Thus, by substituting the suboptimal spectrum sensor operating policy $\pi_{\delta}^{\star}$ into the optimal OSA policy structure in (\ref{OptimalsolutionStructure}),  we  obtain a {\it suboptimal OSA policy} for both (P2-S-1) and (P2-S) as
\begin{equation}
  \left\{
   \begin{array}{l}
  \!\! \delta_a^{\star}(t)=\delta_{a,\mathcal{M}}^{L}(t)+\psi(t)\times (\delta_{a,\mathcal{M}}^{U}(t)-\delta_{a,\mathcal{M}}^{L}(t)), \\
  \!\! ~~~~~~~~~\textrm{~where~}\psi(t)\in[0,1],\\
  \!\!\epsilon_a^{\star}(t) ~\textrm{is on the best ROC curve}~\textrm{corresponding to} ~\delta_a^{\star}(t),\\
   \!\!f_a^{*}(0,t) = 0,  \\
   \!\!f_a^{*}(1,t) = 1,  
   \end{array}
  \right.  \label{pi_1}
\end{equation}
Since $\pi_{\delta}^{\star}$ satisfies the LPUT constraint in (P2-S-2),  the suboptimal OSA policy in (\ref{pi_1}) satisfies the LPUT constraint in (P2-S-1) and thus (P2-S) with equality.

\subsection{Multi-Channel Case}
At last, we  consider  the general case with $N>1$ for (P2). In this case, although the spectrum sensing policy $\pi_s$ generally  depends on the sensor operating policy $\pi_{\delta}$ and the spectrum access policy $\pi_c$, a suboptimal policy for (P2) can be obtained by separately designing $\pi_s$ from  $\pi_{\delta}$ and $\pi_c$, i.e., applying a separation principle. Based on the results for $N=1$, a suboptimal policy for $N>1$ is proposed as follows.
\begin{itemize}
 \item Step 1: On each channel $n\in\mathbb{A}_S$, the SU selects the sensor operating point and the spectrum access probabilities in slot $t$ according to  $\pi_{\delta}^{\star}$ and $\pi_{c}^{*}$, as  given in (\ref{pi_1}).
\item Step 2:  Apply $\pi_{\delta}^{\star}$ and $\pi_{c}^{*}$ to obtain the SU's sensing policy $\pi^{\star}_{s}$, which determines the channel to be sensed in slot $t$, by solving the following unconstrained POMDP:
 \begin{equation}
  \pi^{\star}_{s}=\mathrm{arg}\mathop{\max}_{\pi_{s}} E_{\pi_{s}}\Big\{\sum_{t=1}^{T} R_S(t)|\mathbf{\Lambda}(1),\pi_{\delta}^{\star}, \pi_c^{*}\Big\}. \label{pi_s_LTUCP}
 \end{equation}
\end{itemize}

The sensor operating policy $\pi_{\delta}^{\star}$ and the spectrum access policy $\pi_{c}^{*}$ in (\ref{pi_1}) are provided such that the PU's  normalized throughput  on channel $n$, $n\in\mathbb{A}_S$, equals the benchmark throughput $\Upsilon_n$, if the SU always selects this channel to sense. When $N>1$, the SU has the option to select one from $N$ channels to sense. Thus, the  PU's normalized throughput on  channel $n$ will be at least $\Upsilon_n$. Hence, the LPUT constraint in (P2) over  each channel is satisfied by  the proposed  policy.

\section{Numerical Results}
In this section, we show  by simulation the SU's  and  PU's  throughput under the  proposed reactive PU model.
We assume an energy detection based spectrum sensor for the SU, where the
background noise and the received PU signal are modeled as independent white
Gaussian processes. Let $M$ be the number of PU signal measurements, and $\eta_n$
be the decision threshold for channel $n\in\mathbb{A}_{S}$.  Let $\kappa_{n,0}^2$ and $\kappa_{n,1}^2$ denote the
power of the noise and received PU's signal on channel $n$, respectively. Under
the Neyman-Pearson (NP) criterion, the PFA and PM in slot $t\in \{1,\ldots,T\}$ are obtained as \cite{SP}:
$$
\delta_n(t)\!=\!\gamma \Big(\frac{M}{2},\frac{\eta_n(t)}{2(\kappa_{n,0}^2\!+\!\kappa_{n,1}^2)}\Big),~ \epsilon_n(t)\!=\!1\!-\!\gamma \Big(\frac{M}{2},\frac{\eta_n(t)}{2\kappa_{n,0}^2}\Big) \label{PFA&PM}
$$
where $\gamma(a,m)=(1/\Gamma(m)) \! \times \! \int_0^a t^{m-1}e^{-t}dt$ is the
incomplete gamma function \cite{Gamma}. The optimal decision threshold $\eta_n^{*}(t)$ in slot $t$ of the
energy detector is chosen  such that $\delta_n (t)=\zeta$, if the   SCCP constraint is adopted,  or $\delta_n (t)=\delta_n^{\star}(t)$, if the  LPUT constraint is adopted. Furthermore, we
set $\kappa_{n,0}^2=0$ dB, $\kappa_{n,1}^2=5$ dB, $\forall n \in \mathbb{A}_{S}$, and $M=30$.

\subsection{Single-Channel Case}
We first study the PU's and the SU's performance in the single-channel case. We set $\alpha_{0}^a=0.1$, $\beta_{0}^a=0.2$, $\alpha_{1}^a=0.9$, and $\beta_{1}^a=0.95$. According to the stationary distribution of the underlying Markov chain, the PU's initial channel access probability in slot $t=1$ is obtained   as  $\frac{1-\beta_{0}^a}{1-\beta_{0}^a+\alpha_{0}^a}=0.889$.
We consider two cases: $\zeta=0.05$ and $\zeta=0.1$.
From (\ref{reference_throughput}), we obtain the PU's benchmark throughput $\Upsilon_a=0.846$ when $\zeta=0.05$, and $\Upsilon_a=0.8$ when $\zeta=0.1$.
In each case, we first show the SU's normalized throughput under the non-reactive PU model and the reactive PU model, both subject to the SCCP constraint. The non-reactive PU model is obtained equivalently by setting $\alpha_{1}^a=\alpha_{0}^a=0.1$ and $\beta_{1}^a=\beta_{0}^a=0.2$ in the reactive PU model. The SU's normalized throughput under the non-reactive PU model is computed based on the SU's optimal OSA policy in \cite{Y.ChenIT08}.
We then focus on the reactive PU model and compare the PU's and the SU's normalized throughput under the SCCP constraint and the LPUT constraint in both cases of $\zeta=0.05$ and $\zeta=0.1$. Specifically, under the SCCP constraint, the SU adopts the optimal sensor operating policy $\pi_{\delta}^{*}$ and the optimal spectrum access policy $\pi_c^{*}$, as shown in Section IV, while under the LPUT constraint, the SU adopts the suboptimal sensor operating policy $\pi_{\delta}^{\star}$ with $\psi(t)=0.8$ and the optimal access policy $\pi_c^{*}$, as shown in Section V.

\begin{figure}
\centering
\DeclareGraphicsExtensions{.eps,.mps,.pdf,.jpg,.png}
\DeclareGraphicsRule{*}{eps}{*}{}
\includegraphics[angle=0, width=0.45\textwidth]{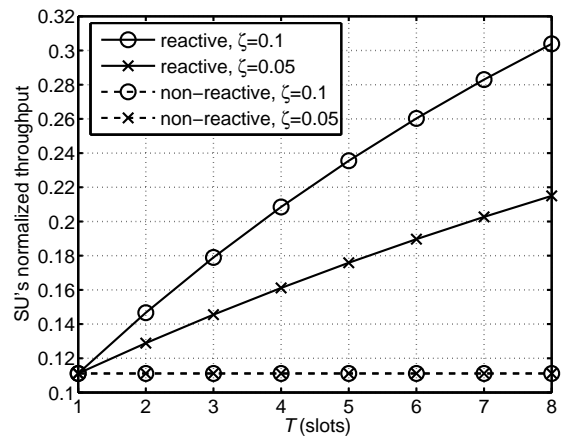} 
\caption{SU's normalized throughput under the SCCP constraint. $N=1$.}
\label{fig:RPU_NRPU_single}
\end{figure}

\begin{figure}
\centering
\DeclareGraphicsExtensions{.eps,.mps,.pdf,.jpg,.png}
\DeclareGraphicsRule{*}{eps}{*}{}
\includegraphics[angle=0, width=0.45\textwidth]{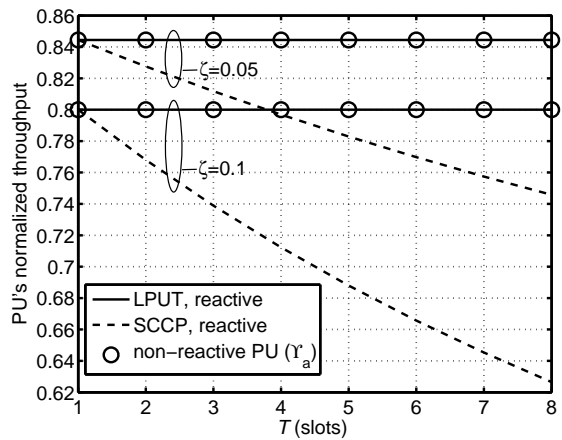}
\caption{PU's normalized throughput. $N=1$.}
\label{fig:PU_single}
\end{figure}

Fig.~\ref{fig:RPU_NRPU_single} shows the SU's normalized throughput under the non-reactive PU model as well as the reactive PU model, by adopting the SCCP constraint.  We observe that the SU achieves higher throughput in the latter than in the former model. This is mainly because that the reactive PU reduces but the non-reactive PU remains the channel access probability after a collision with the SU occurs. Since the non-reactive PU has the same channel access probability, when $N = 1$, the expected channel access opportunities are unchanged over time for the SU. As a result, the SU's throughput is a constant  under the non-reactive PU model. In addition, we observe that the SU's throughput under the non-reactive PU model remains the same in both cases of $\zeta=0.05$ and $\zeta=0.1$. However, the SU's throughput under the reactive PU model is higher in the case of $\zeta=0.1$ than that in the case of $\zeta=0.05$ due to the more relaxed SCCP constraint.

Fig.~\ref{fig:PU_single} compares the PU's normalized throughput with the benchmark throughput. According to Section IV-B, the benchmark throughput, which remains as a constant over time in both cases of $\zeta=0.05$ and $\zeta=0.1$, is actually the non-reactive PU's normalized throughput under the SCCP constraint in the single-channel case. Thus, the non-reactive PU is effectively protected by the SCCP constraint. However, under the reactive PU model, it is observed that if the SU adopts the SCCP constraint, the PU's normalized throughput is lower than the benchmark throughput $\Upsilon_a$ in both cases of $\zeta=0.05$ and $\zeta=0.1$ if $T>1$; and thus the PU is not protected properly, which is in accordance with Proposition \ref{proposition: ineffectiveness of SCCP}. Furthermore, the PU's throughput loss is more substantial in the case of $\zeta=0.1$ than  in the case of $\zeta=0.05$, since the PU allows more collisions  when $\zeta$ is larger. On the other hand, if the SU adopts the LPUT constraint, we
observe that the PU's normalized throughput is equal to the benchmark throughput in both cases, and thus the PU is protected as expected.

\begin{figure}
\centering
\DeclareGraphicsExtensions{.eps,.mps,.pdf,.jpg,.png}
\DeclareGraphicsRule{*}{eps}{*}{}
\includegraphics[angle=0, width=0.45\textwidth]{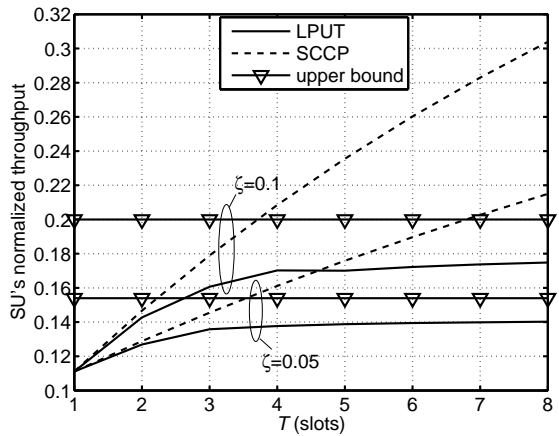}
\caption{SU's normalized throughput under the reactive PU model. $N=1$.}
\label{fig:SU_single}
\end{figure}

Fig.~\ref{fig:SU_single} shows the corresponding SU's normalized throughput under the reactive PU model.
We observe that the SCCP constraint leads to a higher SU throughput than the LPUT constraint in both cases of $\zeta=0.05$ and $\zeta=0.1$. This is because the SU can exploit the PU's reaction to get more throughput under the SCCP constraint.
It is also observed from Fig.~\ref{fig:SU_single} that, unlike the case under the SCCP constraint,  the SU's  throughput under the LPUT constraint does not always increase over $T$. For example, the SU's throughput obtained  with $T=5$ is lower than that with $T=4$.
In addition, by comparing the SU's normalized throughput in both cases, we observe that the SU achieves higher throughput with $\zeta=0.1$ as compared to $\zeta=0.05$. This is consistent with the PU's higher throughput loss when $\zeta=0.1$ as shown in  Fig.~\ref{fig:PU_single}.
To evaluate the performance of the proposed suboptimal OSA policy for the SU, in the following, we propose an \emph{upper bound} of the SU's normalized throughput for the single-channel case, under the constraint that the PU must achieve the benchmark throughput. By noticing the fact that unit throughput is the maximum normalized throughput that a channel
can provide, the upper bound is given by the difference between the unit throughput and the PU's benchmark throughput on the channel. Since the SU's throughput loss due to sensing errors is not considered, the upper bound is higher than the SU's maximum normalized throughput when the PU achieves the benchmark throughput. As shown in  Fig.~\ref{fig:SU_single}, we compare the SU's normalized throughput with the upper bound in both cases of $\zeta=0.05$ and $\zeta=0.1$ under both the SCCP and LPUT constraints. It is observed that the SU's normalized throughput under the suboptimal policy is always lower than the upper bound under the LPUT constraint. However, under the SCCP constraint, after $T=3$ in both cases of $\zeta=0.05$ and $\zeta=0.1$, the SU's normalized throughput becomes higher than the \emph{upper bound}, since the PU's achievable throughput deviates from the benchmark throughput.

\begin{figure}
\centering
\DeclareGraphicsExtensions{.eps,.mps,.pdf,.jpg,.png}
\DeclareGraphicsRule{*}{eps}{*}{}
\includegraphics[angle=0, width=0.45\textwidth]{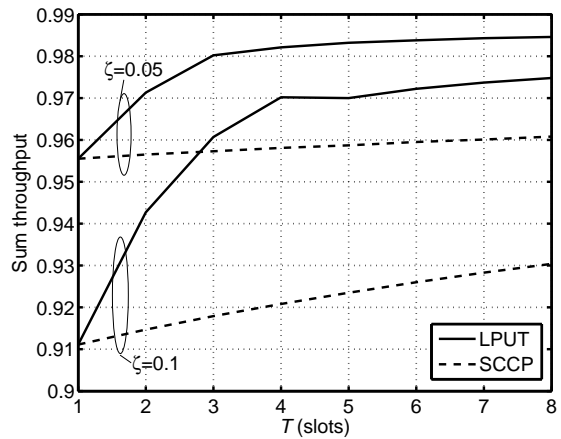}
\caption{Sum throughput under the reactive PU model. $N=1$.}
\label{fig:Sum_single}
\end{figure}

Fig.~\ref{fig:Sum_single} compares the sum of SU's and PU's normalized throughput under the SCCP and LPUT constraints in both cases of $\zeta=0.05$ and $\zeta=0.1$, where the reactive PU model is considered.
It is observed from Fig.~\ref{fig:Sum_single} that,  in both cases, the LPUT constraint leads to a higher sum-throughput than the SCCP constraint. It is worth pointing out that, although the sum-throughputs under the SCCP and LPUT constraints are close, the individual portions of the SU's and PU's throughput, as shown in Fig.~\ref{fig:PU_single} and Fig.~\ref{fig:SU_single}, are very different under these two constraints.

\subsection{Multi-Channel Case}
Next, we consider the multi-channel case by assuming $N=3$. Since the performances of SU's and PU's throughput under the non-reactive PU model in the multi-channel case are similar to those in the single-channel case, respectively, we only consider the reactive PU model in this case. The reactive PU model is given by four vectors  $\pmb{\alpha_0}=(0.1, 0.1, 0.05)$, $\pmb{\beta_0}=(0.1, 0.2, 0.6)$, $\pmb{\alpha_1}=(0.9, 0.9, 0.9)$, and $\pmb{\beta_1}=(0.95, 0.95, 0.95)$. According to the stationary distribution of the underlying Markov chain, the PU's initial channel access probabilities in slot $t=1$ are obtained   as $(0.9,0.889,0.889)$.
Since the results of the SU's and the PU's throughput with $\zeta=0.1$ are similar to those with $\zeta=0.05$, we only show the case of $\zeta=0.05$ under the reactive PU model. From (\ref{reference_throughput}), the  PU's benchmark throughput is obtained as $\pmb{\Upsilon}=(\Upsilon_1,\Upsilon_2,\Upsilon_3)=(0.855,0.846,0.846)$ with $\zeta=0.05$. Similar to the single-channel case, the SU adopts $\pi_{\delta}^{\star}$ with $\psi(t)=0.8$.

\begin{figure}
\centering
\DeclareGraphicsExtensions{.eps,.mps,.pdf,.jpg,.png}
\DeclareGraphicsRule{*}{eps}{*}{}
\includegraphics[angle=0, width=0.45\textwidth]{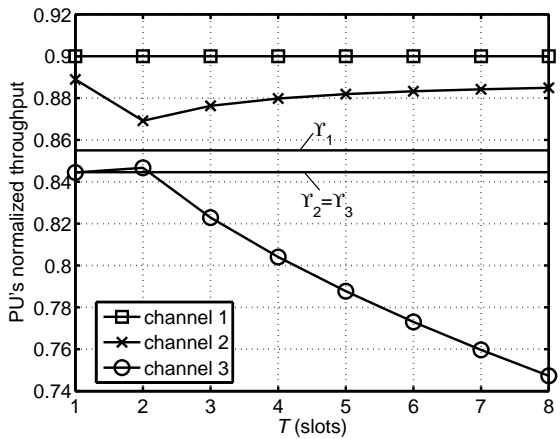}
\caption{PU's normalized throughput under the SCCP constraint. $N=3$.}
\label{fig:PU_3channels_SCCP}
\end{figure}

\begin{figure}
\centering
\DeclareGraphicsExtensions{.eps,.mps,.pdf,.jpg,.png}
\DeclareGraphicsRule{*}{eps}{*}{}
\includegraphics[angle=0, width=0.45\textwidth]{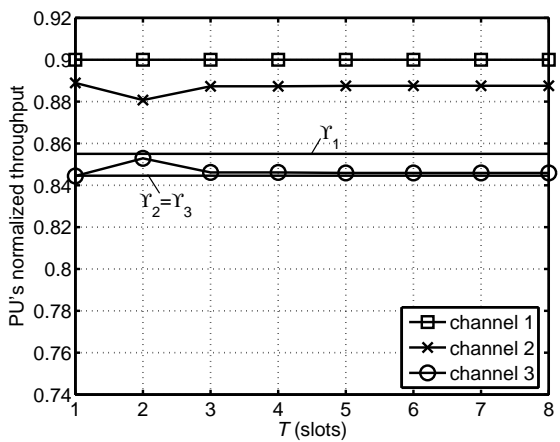}
\caption{PU's normalized throughput under the LPUT constraint. $N=3$.}
\label{fig:PU_3_channels_LPUT}
\end{figure}

Fig.~\ref{fig:PU_3channels_SCCP} shows the PU's normalized throughput under the SCCP constraint. It is observed that the PU's  normalized throughput on channel 1 is higher than the benchmark throughput $\Upsilon_1$  and remains as $0.9$ over $T$, which is the PU's throughput with the absence of SU, while the PU's normalized throughput on channel 2 and channel 3 vary over $T$. This indicates that the SU only selects channel 2 and channel 3 to access in this example, since the SU is able to achieve more reward on these two channels, where the PUs have lower  initial channel access probabilities than on  channel 1.
We also observe that the PU's throughput on channel 2 is always larger than its benchmark throughput $\Upsilon_2$. However, when $T>2$, the PU's throughput on channel 3 decreases over $T$ and becomes lower than the benchmark throughput $\Upsilon_3$. This indicates that the SU selects channel 3 to access in most of the slots. Thus, the PU in  channel~2 is protected as expected, while the PU on channel 3 is not protected properly.
Hence, as we discussed in Section~IV,  the SCCP constraint is in general not able to provide effective protection to all the reactive PUs when $N>1$.

Fig.~\ref{fig:PU_3_channels_LPUT} shows the PU's normalized throughput under the LPUT constraint. Similar to Fig.~\ref{fig:PU_3channels_SCCP} under the SCCP constraint, the PUs on channels 1 and  2 are both properly protected by the LPUT constraint, since their normalized throughput are larger than their respective benchmark throughput $\Upsilon_1$ and $\Upsilon_2$, respectively. Different from  Fig.~\ref{fig:PU_3channels_SCCP}, where the PU's normalized throughput on channel 3 is not guaranteed to meet the benchmark throughput $\Upsilon_3$, we observe from Fig.~\ref{fig:PU_3_channels_LPUT} that the throughput under the LPUT constraint is higher than $\Upsilon_3$, i.e., the PU on channel 3 is protected properly. This shows that the LPUT constraint provides effective protection to all the reactive PUs.

\begin{figure}
\centering
\DeclareGraphicsExtensions{.eps,.mps,.pdf,.jpg,.png}
\DeclareGraphicsRule{*}{eps}{*}{}
\includegraphics[angle=0, width=0.45\textwidth]{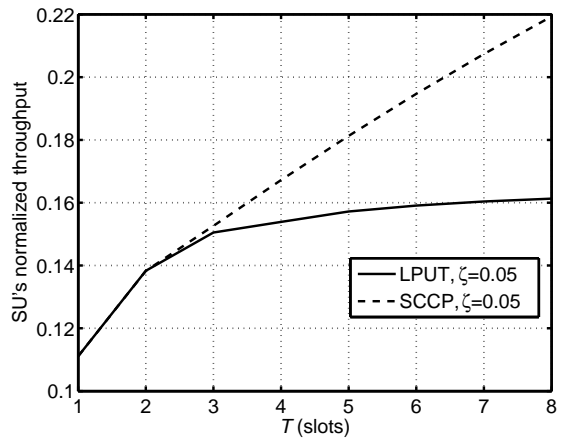}
\caption{SU's normalized throughput under the LPUT constraint. $N=3$.}
\label{fig:SU_3_channels_LPUT}
\end{figure}

\begin{figure}
\centering
\DeclareGraphicsExtensions{.eps,.mps,.pdf,.jpg,.png}
\DeclareGraphicsRule{*}{eps}{*}{}
\includegraphics[angle=0, width=0.45\textwidth]{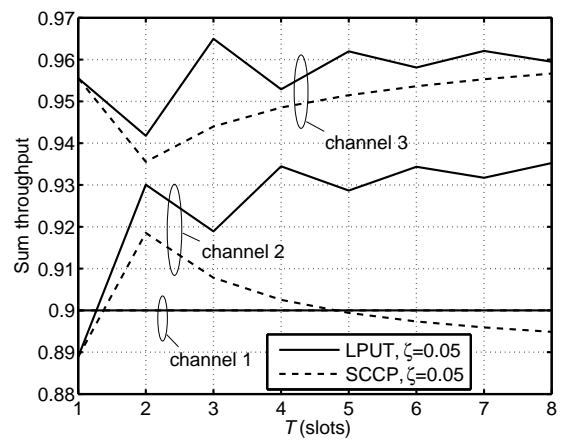}
\caption{Sum throughput under the reactive PU model. $N=3$.}
\label{fig:Sum_multiple}
\end{figure}

Fig.~\ref{fig:SU_3_channels_LPUT} compares the SU's normalized throughput under the SCCP and LPUT constraints. Similar to the single-channel case shown in Fig.~\ref{fig:SU_single}, the SU achieves higher throughput under the SCCP constraint than under the LPUT constraint. Compared with the single-channel case of $\zeta=0.05$ in Fig.~\ref{fig:SU_single}, we observe that the SU achieves higher throughput in the
multi-channel case under both the SCCP and LPUT constraints. This is because when $N>1$, the SU has more flexibility in selecting channels that are more likely to be unoccupied to access.

Fig.~\ref{fig:Sum_multiple} compares the sum of SU's and PU's normalized throughput on each channel under the SCCP and LPUT constraints. It is already shown in Fig.~\ref{fig:PU_3channels_SCCP} (SCCP constraint) and Fig.~\ref{fig:PU_3_channels_LPUT} (LPUT constraint) that the SU only accesses channel 2 or channel 3 and does not access channel 1. Therefore, we observe from Fig.~\ref{fig:Sum_multiple} that the sum-throughput on channel 1 under both constraints remains unchanged over time, which is equal to the PU's normalized throughput on channel 1. It is also observed from Fig.~\ref{fig:Sum_multiple} that, on channel 2 and channel 3, the LPUT constraint leads to a higher   sum-throughput than the SCCP constraint, which is similar to the single-channel case, as shown in Fig.~\ref{fig:Sum_single}.

\section{Conclusion}
In this paper, we studied a practical multi-channel CR network overlaid with reactive PUs. We proposed a new  channel access model for the reactive PU, in which the probability for the PU to access a particular channel is related to the SU's past access decisions. Under this model, we formulated the optimal OSA design for the SU's throughput maximization as a constrained POMDP problem. We considered both SCCP and LPUT constraints to protect the reactive PU's transmission. For the SCCP constraint, we developed the optimal OSA policy via a separation principle. 
For the LPUT constraint, we developed the structure of the optimal OSA policy. In order to reduce the computational complexity, we converted the POMDP into an equivalent MDP with deterministic state transitions. With the reformulated LPUT constraint, we proposed a suboptimal policy of lower complexity. It is shown that the proposed  policy guarantees PU's throughput for both single-channel and multi-channel cases.

\appendices

\section{Proof of Theorem \ref{theorem: separation principle}}\label{appendix:proof 1}
First, we present the following lemma.
\begin{lemma} \label{lemma: in A}
$(\epsilon_{a}^{*}(t), \delta_{a}^{*}(t))$ and $(f_{a}^{*}(0,t),f_{a}^{*}(1,t))$ given in (\ref{Optimalsolution}) are the optimal solutions to the   problem
\begin{align}
&\mathop{\max}_{(\epsilon_{a}(t), \delta_{a}(t)) \in\mathbb{A}_{\delta}(a(t)), \atop (f_{a}(0,t),f_{a}(1,t)) \in [0,1]^{2}}
C_{1} \times \mu_a(t)+C_{2} \times g_a(t)+C_3 \notag \\
&~~~~~~~~~s.t. ~~~~~~~\mu_a(t)\leq \zeta.  \notag
\end{align}
where $C_1$, $C_2$, and $C_3$ are constants with $C_1 \geq 0$, $C_2 \geq 0$,  $g_a(t)$ is given in (\ref{g_a}), and $\mu_a(t)$ is given in (\ref{mu_a}).
\end{lemma}

\begin{proof}
In \cite{Y.ChenIT08}, the authors proved that (\ref{Optimalsolution}) is the optimal solution to the     problem
\begin{align}
&\mathop{\max}_{(\epsilon_{a}(t), \delta_{a}(t)) \in\mathbb{A}_{\delta}(a(t)), \atop (f_{a}(0,t),f_{a}(1,t)) \in [0,1]^{2}}
g_a(t) \notag \\
&~~~~~~~~~s.t. ~~\mu_a(t)\leq \zeta. \notag
\end{align}
By applying (\ref{Optimalsolution}), the constraint function $\mu_a(t)$ achieves  $\zeta$ with equality. Since  $C_1 \geq 0$, $C_2 \geq 0$, when the solution is given by (\ref{Optimalsolution}), the objective function is maximized. Lemma \ref{lemma: in A} thus follows.
\end{proof}

With Lemma A.1, we now prove the separation principle. By mathematical induction, in the following, we prove that  (\ref{Optimalsolution}) is the SU's optimal action decision for (P1) in each slot $t$, $1\leq t \leq T$, and with (\ref{Optimalsolution}), the value function $V_{t}(\mathbf{\Lambda}(t))$ in slot $t$ is of the following form:
\begin{align}
V_{t}(\mathbf{\Lambda}(t))=& D_t \times \big(\lambda_{a0}(t)+\lambda_{a2}(t)\big) +F_t \times \lambda_{a1}(t) \nonumber \\
+&H_t \times \lambda_{a3}(t) + z_t, ~~1\leq t\leq T \label{V_tForm}
\end{align}
where   $D_t \geq 0$, $H_t\geq F_t \geq 0$, and $z_t$ is independent of $\mathbf{\Lambda}(t)$.

Specifically, in the last slot $t=T$,  suppose that channel $a$ is selected to sense. From (\ref{Q_T}), we then have $Q_{T}(\mathbf{\Lambda}(T)|A)=(\lambda_{a1}(T)+ \lambda_{a3}(T))\times g_a(T)$.
Since $\lambda_{a1}(T)+ \lambda_{a3}(T)\geq 0$, from Lemma A.1, it is easy to verify  that (\ref{Optimalsolution}) is the SU's optimal action decision in slot $T$, such that $Q_{T}(\mathbf{\Lambda}(T)|A(T))$ is maximized subject to  the SCCP constraint.
Applying (\ref{Optimalsolution}) to $Q_{T}(\mathbf{\Lambda}(T)|A)$, we obtain $V_{T}(\mathbf{\Lambda}(T))=(\lambda_{a1}(T)+ \lambda_{a3}(T))\times (1-\epsilon_a^{*}(T))$, which follows the form given in (\ref{V_tForm}).

Now suppose that in slot $t$, channel $a$ is selected to sense, and  (\ref{Optimalsolution}) is the SU's optimal action decision in slot $t$ with value function $V_{t}(\mathbf{\Lambda}(t))$ given by (\ref{V_tForm}). Next, supposing  $a(t-1)=n$  in slot $t-1$, given $\mathbf{\Lambda}(t-1)$, we derive the optimal action decision in slot $t-1$ and the value function $V_{t-1}(\mathbf{\Lambda}(t-1))$ in the following two cases:
\begin{itemize}
\item \emph{Case 1: $n \neq a$.}  According to (\ref{updatedBM2}) and (\ref{Q_t}) and after some algebra, we obtain
$Q_{t-1}(\mathbf{\Lambda}(t-1)|A(t-1))=\big(\epsilon_{n}(t-1)f_{n}(0,t-1)+(1-\epsilon_{n}(t-1))f_{n}(1,t-1)\big)\times \big(\lambda_{n1}(t-1)+\lambda_{n3}(t-1)\big)+V_{t}(\mathbf{\Lambda}(t))$. Since  from (\ref{updatedBM2}), $\mathbf{\Lambda}(t)$ is independent of the SU's action $A(t-1)$, $V_{t}(\mathbf{\Lambda}(t))$ is treated as  a constant.
Then according to Lemma A.1, (\ref{Optimalsolution}) is the SU's optimal action  to maximize $Q_{t-1}(\mathbf{\Lambda}(t-1)|A(t-1))$ subject to the SCCP constraint. Applying (\ref{Optimalsolution}) to  $Q_{t-1}(\mathbf{\Lambda}(t-1)|A(t-1))$ yields that $V_{t-1}(\mathbf{\Lambda}(t-1))=(1-\epsilon_a^{*}(t-1))\times (\lambda_{n1}(t-1)+\lambda_{n3}(t-1))+V_{t}(\mathbf{\Lambda}(t))$. Clearly,  $V_{t-1}(\mathbf{\Lambda}(t-1))$  follows the form given in (\ref{V_tForm}).

\item \emph{Case 2: $n=a$.} Similarly to {\it Case 1}, according to (\ref{updatedBM1}) and (\ref{Q_t}), we can obtain the expression of $Q_{t-1}(\mathbf{\Lambda}(t-1)|A(t-1))$, in which $V_{t}(\mathbf{\Lambda}(t))$ is related to the SU's action $A(t-1)$. By   adopting the same method  as in {\it Case 1}, it is easy to verify that  (\ref{Optimalsolution}) is  the SU's optimal action  and the resultant $V_{t-1}(\mathbf{\Lambda}(t-1))$ follows the form given in (\ref{V_tForm}), 
with  $D_{t-1}=\zeta[H_t\alpha_1^a-D_t\alpha_1^a-F_t\alpha_0^a+D_t\alpha_0^a]+F_t\alpha_0^a+D_t(1-\alpha_0^a)\geq 0$, $F_{t-1}=F_t\beta_1^a+D_t-D_t\beta_1^a+1-\epsilon_a^{*}(t-1)\geq 0$,  $H_{t-1}=H_t\beta_1^a+D_t-D_t\beta_1^a+1-\epsilon_a^{*}(t-1)\geq 0$ and $z_{t-1}=z_t$.  Obviously, $H_{t-1}\geq F_{t-1}$.
\end{itemize}

Hence, (\ref{Optimalsolution}) is the SU's optimal action in each slot under the SCCP constraint.
From (\ref{Optimalsolution}), the optimal sensor operating point and the optimal access probabilities are constant and  independent from the channel selected to sense.
As a result, we can separately design the optimal spectrum sensing policy  as shown in Theorem~\ref{theorem: separation principle}, without loss of  optimality. Theorem \ref{theorem: separation principle} is thus proved.

\section{Proof of Proposition \ref{proposition: ineffectiveness of SCCP}}\label{appendix:proof 2}
For the case of $N=1$ with $n=a$, there is only one PU and SU pair sharing one channel and the SU's polices are reduced to be $\pi_{\delta}$ and $\pi_{c}$.
Given the current belief state $\mathbf{\Lambda}(t)$, we use $G_t(\mathbf{\Lambda}(t)|\pi)$ to denote the PU's  throughput on channel $a$ from slot $t$ to slot $T$ under the SU's policy  $\pi=\{\pi_{\delta}, \pi_{c}\}$. Similar to $Q_t(\mathbf{\Lambda}(t)|A)$ for the SU in (\ref{Q_t}) and (\ref{Q_T}), from (\ref{PU_IMMR_exp}) and (\ref{eq: PU_normal_Throughput}) and with the fact that $\sigma_a(t)=\mu_a(t)$, $t\in[0,1]$, for $N=1$, we have
\begin{align}
&G_t(\mathbf{\Lambda}(t)|\pi)\!=\!(\lambda_{a0}(t)\!+\!\lambda_{a2}(t))  (1\!-\!\mu_a(t)) \nonumber \\
&~~~~~~~~~~~~~~~~+\sum_{i=0}^3 \sum_{k=0}^{1}\lambda_{ai}U_{A}(k|i)G_{t\!+\!1}(\mathbf{\Lambda}(t\!+\!1)|\pi), \nonumber \\
&~~~~~~~~~~~~~~~~~~~~~~~~~~~~~~~~~~~~~~~~~~~~~~~~1\!\leq \!t\! \leq\! T\!-\!1,  \label{G_t}  \\
&G_T(\mathbf{\Lambda}(T)|\pi)\!=\!(\lambda_{a0}(T)\!+\!\lambda_{a2}(T))  (1\!-\!\mu_a(T)), ~~t\!=\! T. \label{G_T}
\end{align}
It is easy to find that $$G_1(\mathbf{\Lambda}(1)|\pi)=E_{\pi}\big\{\sum_{t=1}^{T} R_{P,a}(t)|\mathbf{\Lambda}(1)\big\}$$
which is the PU's  throughput on channel $a$ over all $T$ slots. Thus,  $R_{P,a}^{o}=G_1(\mathbf{\Lambda}(1)|\pi)/T$.

From (\ref{G_t}) and (\ref{G_T}), to find $G_1(\mathbf{\Lambda}(1)|\pi)$, we need to compute $G_t(\mathbf{\Lambda}(t)|\pi)$ for all $t\!\in\! \{1,\ldots,T\}$.
We consider PU's throughput from slot $t$ to slot $T$  under two cases. One is with the SU's optimal policy for (P1) under the reactive PU model. The other is with the SU's optimal policy in \cite{Y.ChenIT08}, which is  under the SCCP constraint and under the non-reactive PU model.
For notational convenience, we denote PU's throughput from slot $t$ to slot $T$ obtained in the former case by $G_t$, and denote that in the latter case by $G_t^{'}$. Note that $G_1^{'}/T=\Upsilon_a$. We take
$G_1^{'}$ as a reference and show that $G_1<G_1^{'}$. Since the proof is similar to that in Appendix A, in the following, we only provide the proof sketch.

Based on mathematical induction and by computing backward in time from (\ref{G_t}) and (\ref{G_T}), it is easy to  find that in  slot $t$, $\forall t\in \{1,\ldots,T\}$ and $T>1$,   $G_t$ is of the following form:
\begin{equation}
 G_t= (\lambda_{a0}(t)+\lambda_{a0}(t))\times q_t+\lambda_{a1}(t)\times w_t+\lambda_{a3}(t)\times m_t, \label{G_uniform}
\end{equation}
where the coefficients $q_t$, $w_t$ and $m_t$ are time-varying and depend on $(\alpha_0^a, \beta_0^a,\alpha_1^a, \beta_1^a)$.
Since the SU's optimal polices under the two cases are the same, given $\mathbf{\Lambda}(1)$, the belief states $\mathbf{\Lambda}(t)$ under two cases in each slot  are also the same. Moreover,  given $\mathbf{\Lambda}(t)$, $G_t^{'}$ has the similar expression as $G_t$ in slot $t$, but with different coefficients, which are denoted by $q_t^{'}$, $w_t^{'}$, and $m_t^{'}$. Note that by reducing $\alpha_1^a$ and $\beta_1^a$ to $\alpha_1^a=\alpha_0^a$ and $\beta_1^a=\beta_0^a$,  the reactive PU model is reduced to the non-reactive counterpart. Thus, by doing so, in each slot, the coefficients of $G_t$, i.e., $q_t$, $w_t$, and $m_t$, are reduced to those of $G_t^{'}$, i.e., $q_t^{'}$, $w_t^{'}$, and $m_t^{'}$, respectively.
Based on mathematical induction, we find that in  slot $t$, $q_t<q_t^{'}$, $w_t\leq w_t^{'}$ and $m_t<m_t^{'}$. Thus, we obtain   $G_t<G_t^{'}$, $\forall t\in \{1,\ldots,T\}$.

Hence, we have  $G_1/T<G_1^{'}/T$, i.e., for (P1) and under the SU's optimal policy, the PU's normalized throughput $R_{P,a}^{o}<\Upsilon_a$. Proposition \ref{proposition: ineffectiveness of SCCP} is thus proved.

\section{Proof of Proposition \ref{proposition: PU_throughput_constraint_with_equality}} \label{appendix:proof 3}
We first study the PU's throughput under the policy $\pi$ and construct  policy $\pi^{'}$ based on $\pi$.
Suppose under the policy $\pi$, the PU's  throughput from slot $t=1$ to slot $t=T-1$ is $E_{\pi}\big\{\sum_{t=1}^{T-1} R_{P,a}(t)|\mathbf{\Lambda}(1)\big\}=\rho_a$, and  the PU's  throughput in the last slot $t=T$ is $u_T$. From (\ref{PU_IMMR_exp}), we have
\begin{equation}
(\lambda_{a0}(T)+\lambda_{a2}(T))\times (1-\mu_a(T))=u_T \label{pi},
\end{equation}
where $\mu_a(T)$ is given in (\ref{mu_a}).
Note that under policy $\pi$, $ E_{\pi}\big\{\sum_{t=1}^{T}\! R_{P,a}(t)|\mathbf{\Lambda}(1)\big\}\!>\!\Upsilon_a \!\times\! T$. Thus, $u_T\!>\!\Upsilon_a \!\times\! T\!-\!\rho_a$.

The policy $\pi^{'}$ is constructed based on $\pi$. We let the decision functions, as described in Section III-C, from slot $t=1$ to slot $t=T-1$ of policy  $\pi^{'}$ be the same as those of policy $\pi$.   We thus have $E_{\pi^{'}}\big\{\sum_{t=1}^{T-1} R_{P,a}(t)|\mathbf{\Lambda}(1)\big\}=\rho_a$. Different from  policy $\pi$, under the policy $\pi^{'}$ in the last slot $t=T$, we let $E_{\pi^{'}}\big\{R_{P,a}(T)|\mathbf{\Lambda}(T)\big\}=\Upsilon_a \times T-\rho_a$  by selecting actions $\delta_a^{'}(T)=1-\frac{\Upsilon_a \times T-\rho_a}{\lambda_{a0}(T)+\lambda_{a2}(T)}$, $\epsilon_a^{'}(T)$ be the one on the optimal ROC curve corresponding to $\delta_a^{'}(T)$, $f_a^{'}(0,T)=0$, and $f_a^{'}(1,T)=1$.
Note that since $u_T>\Upsilon_a \times T-\rho_a$, from (\ref{pi}), it is clear that $0\leq \delta_a^{'}(T)<1$.
That is, we select feasible actions such that $(\lambda_{a0}(T)+\lambda_{a2}(T))\times (1-\mu_a^{'}(T))=\Upsilon_a \times T-\rho_a \label{pi'}$.

Next,  we show that $$E_{\pi}\big\{ \sum_{t=1}^{T}R_S(t)|\mathbf{\Lambda}(1)\big\}< E_{\pi^{'}}\big\{ \sum_{t=1}^{T}R_S(t)|\mathbf{\Lambda}(1)\big\}.$$ Since the decision functions under policies $\pi$ and $\pi^{'}$ are the same from slot $t=1$ to slot $t=T-1$,  the following equation holds for  the SU's expected throughput from slot $t=1$ to slot $t=T-1$:
\begin{equation}
E_{\pi}\big\{ \sum_{t=1}^{T-1}R_S(t)|\mathbf{\Lambda}(1)\big\}= E_{\pi^{'}}\big\{ \sum_{t=1}^{T-1}R_S(t)|\mathbf{\Lambda}(1)\big\}. \label{T-1_pi_pi'}
\end{equation}
Thus, to compare the SU's expected throughput over $T$ slots under policies $\pi$ and $\pi^{'}$, we only need to compare the SU's expected rewards in the last slot $t=T$. It is  easy to find that the belief states $\mathbf{\Lambda}(T)$ in slot $T$ under both polices are the same.
In the following, we  compute an upper bound on the SU's  expected reward in slot $t=T$ under the policy $\pi$, which is denoted as $E_{\pi}^{U}\big\{R_{S}(T)|\mathbf{\Lambda}(T)\big\}$ with $E_{\pi}^{U}\big\{R_{S}(T)|\mathbf{\Lambda}(T)\big\}\geq E_{\pi}\big\{R_{S}(T)|\mathbf{\Lambda}(T)\big\}$, and show that $ E_{\pi^{'}}\big\{R_{S}(T)|\mathbf{\Lambda}(T)\big\} > E_{\pi}^{U}\big\{R_{S}(T)|\mathbf{\Lambda}(T)\big\}$.
Denote the SU's actions in slot $t=T$ that achieve $E_{\pi}^{U}\big\{R_{S}(T)|\mathbf{\Lambda}(T)\big\}$ under the constraint  given in (\ref{pi})  by $\delta_a^{U}(T)$, $\epsilon_a^{U}(T)$, $f_a^{U}(0,T)$, and $f_a^{U}(1,T)$. To find these actions, we need to solve an optimization problem, which is to maximize $(\lambda_{a1}(T)+\lambda_{a3}(T)) \times g_a^{U}(T)$ subject to $\mu_a^{U}(T)=1-\frac{u_T}{\lambda_{a0}(T)+\lambda_{a2}(T)}$. According to 
Lemma A.1, it is easy to find  that $(\delta_a^{U}(T),\epsilon_a^{U}(T) )$ is on the optimal ROC curve with  $\delta_a^{U}(T)\!=\!1\!-\!\frac{u_T}{\lambda_{a0}(T)\!+\!\lambda_{a2}(T)}$,  $f_a^{U}(0,T)\!=\!0$, and $f_a^{U}(1,T)\!=\!1$.
Since  $u_T\!>\!\Upsilon_a \!\times\! T\!-\!\rho_a$, thus $\delta_a^{'}(T)\!>\!\delta_a^{U}(T)$. Correspondingly,  $\epsilon_a^{'}(T)\!<\!\epsilon_a^{U}(T)$.
From (\ref{g_a}), it is obvious that $g_a^{'}(T)\!>\!g_a^{U}(T)$. Thus,  given  $\mathbf{\Lambda}(T)$ and from (\ref{SU_expect_R}), we have $E_{\pi'}\big\{R_{S}(T)|\mathbf{\Lambda}(T)\big\}\!>\!E_{\pi}^{U}\big\{R_{S}(T)|\mathbf{\Lambda}(T)\big\}\!\geq\! E_{\pi}\big\{R_{S}(T)|\mathbf{\Lambda}(T)\big\}$. From (\ref{T-1_pi_pi'}),  we have $E_{\pi^{'}}\big\{ \sum_{t=1}^{T}\!R_S(t)|\mathbf{\Lambda}(1)\big\}\!>\!E_{\pi}\big\{ \sum_{t=1}^{T}\!R_S(t)|\mathbf{\Lambda}(1)\big\}$. Proposition \ref{proposition: PU_throughput_constraint_with_equality} is thus proved.

\section{Proof of Proposition \ref{proposition: optimal structure}} \label{appendix:proof 4}
Suppose that $\pi_{\delta}^{*}$ and $\pi_{c}^{*}$ are the SU's optimal policies for (P2-S-1). Denote  $c^{*}(t)$, $1\leq t \leq T$, as the resultant PU's  throughput in slot $t$ under  $\pi_{\delta}^{*}$ and $\pi_{c}^{*}$.  That is, given $\mathbf{\Lambda}(t)$ in slot $t$, we have $E_{\pi_{\delta}^{*}, \pi_{c}^{*}}\big\{R_{P,a}(t)|\mathbf{\Lambda}(t)\big\}=c^{*}(t)$ and $\sum_{t=1}^{T}c^{*}(t)=\Upsilon_a\times T$.
Thus, if we have found $c^{*}(t)$, $1\leq t \leq T$,  we can adopt the following  short-term protection for PU's transmission for (P2-S-1), without loss of optimality:
\begin{equation}
 E_{\pi_{\delta}, \pi_{c}}\big\{R_{P,a}(t)|\mathbf{\Lambda}(t)\big\}=c^{*}(t), ~\forall t \in \{1,\ldots,T\}.  \label{eq: short_constrain}
\end{equation}
Then from  (\ref{PU_IMMR_exp}) and (\ref{eq: short_constrain}), under the availability assumption of $c^{*}(t)$, (P2-S-1) is  equivalent to
\begin{align}
\mathrm{(\textrm{D-1})}:\mathop{\mathrm{max.}}_{\pi_{\delta},\pi_{c}} &~
E_{\pi_{\delta},\pi_{c}}\big\{ \sum_{t=1}^{T}R_S(t)|\mathbf{\Lambda}(1)\big\}  \nonumber \\
\mathrm{s.t.}  &~\sigma_a(t)=1-\frac{c^{*}(t)}{P\{\emph{I}(a,t)=0\}}, ~\forall t \in \{1,\ldots,T\}.   \nonumber
\end{align}
Since the formulation is similar to  (P1) with $N=1$, from Appendix A,  the optimal solutions for (D-1) are
\begin{equation}
  \left\{
   \begin{array}{l}
   \!\!\delta_a^{*}(t) =1-\frac{c^{*}(t)}{P\{\emph{I}(a,t)=0\}},\\  
  \!\!\epsilon_a^{*}(t) ~\textrm{is on the optimal ROC curve corresponding }\\
  \!\!~~~~~~~\textrm{to} ~ \delta_a^{*}(t), \\
   \!\!f_a^{*}(0,t) = 0,  \\
   \!\!f_a^{*}(1,t) = 1,  
   \end{array}
  \right. \label{eq: appendix_D}
\end{equation}
Since (D-1) is equivalent to (P2-S-1), (\ref{eq: appendix_D}) is also the optimal solutions for (P2-S-1).
Furthermore, with a proof similar to that in Appendix B, it is easy to show that if $\frac{c^{*}(t)}{P\{\emph{I}(a,t)=0\}}$ is a constant over time $t$, the resultant PU's throughput over $T$ slots will be smaller than $\Upsilon_a\times T$, which is contrary to the fact that  $\sum_{t=1}^{T}c^{*}(t)=\Upsilon_a\times T$. Thus, $\frac{c^{*}(t)}{P\{\emph{I}(a,t)=0\}}$ is not a constant over time $t$. Hence,  the optimal PM decision $\delta_a^{*}(t)$ is  time-varying and needs to be adaptively selected based on $\mathbf{\Lambda}(t)$ over time.
Proposition \ref{proposition: optimal structure} is thus proved.

\section{Proof of Proposition \ref{proposition: optimal transformed MDP}} \label{appendix:proof 5}
Firstly, we introduce some new notations for the POMDP.
According to the complexity analysis for (P2-S-2) in Section V-A, given $\mathbf{\Lambda}(1)$  and the SU's POMDP policy $\pi_{\delta}$ for (P2-S-2), $2^{t-1}$ possible belief states could occur with non-zero probability in slot $t$, $\forall t\in\{1,\ldots,T\}$.
We use $\mathbf{\Lambda}^{b}(t)=(\lambda_{a0}^b(t),\lambda_{a1}^b(t),\lambda_{a2}^b(t),\lambda_{a3}^b(t))$, $b\in\{1,\ldots,2^{t-1}\}$ to denote these possible belief states in slot $t$ and use $h^b(t)$ to denote the occurrence probability of the belief state $\mathbf{\Lambda}^{b}(t)$,
where $h^b(t)$ is determined by the SU's action decision history and observation history in the previous $t-1$ slots and $\sum_{b=1}^{2^{t-1}}h^b(t)=1$.
We denote the SU's PM decision on channel $a$ for  belief state $\mathbf{\Lambda}^{b}(t)$ as $\delta_a^b(t)$.
Next, under the MDP policy $\pi_{\delta,\mathcal{M}}^{\star}$ and the POMDP policy $\pi_{\delta}^{\star}$,  we  give the following lemma, based on which, Proposition \ref{proposition: optimal transformed MDP} can be proved.
\begin{lemma} \label{lemma: state relationship}
Given $\mathbf{\Lambda}(1)=\mathbf{\Omega}(1)$,  the relationship between the POMDP belief states and the MDP state in each slot is
\begin{equation}
\omega_{aj}(t)\!=\!\sum_{b=1}^{2^{t-1}} h^b(t) \times\lambda_{aj}^b(t),~\forall j\!\in \!\mathbb{C}_{S},~\forall t\!\in\!\{1,\ldots,T\}. \label{State_Relation}
\end{equation}
\end{lemma}
\begin{proof}
We use mathematical induction to prove this lemma. Since $\mathbf{\Omega}(1)=\mathbf{\Lambda}(1)$,
(\ref{State_Relation}) holds when $t=1$.
Suppose (\ref{State_Relation}) holds in slot $t$, $t>1$, by applying (\ref{kinU}) to compute $h^b(t+1)$ and applying (\ref{updatedBM1})
and (\ref{MDP_TP}) to update the POMDP belief state and the MDP state,
respectively, after some algebra, we find that   (\ref{State_Relation}) still holds in slot $t+1$.
Lemma~\ref{lemma: state relationship} is thus proved.
\end{proof}

We  now  prove Proposition \ref{proposition: optimal transformed MDP}.
By computing over all the possible belief states in slot $t$, the PU's  throughput under the POMDP policy $\pi_{\delta,\mathcal{M}}^{\star}$ in slot $t$ is
\begin{align}
&E_{\pi_{\delta,\mathcal{M}}^{\star}}\big\{R_{P,a}(t)|\mathbf{\Lambda}(1)\big\} \nonumber \\
&= \sum_{b=1}^{2^{t-1}} h^b(t) \times (\lambda_{a0}^b(t)+\lambda_{a2}^b(t))\times (1-\delta_{a,\mathcal{M}}^{\star}(t)). \label{PU_IMMR_b}
\end{align}
Under the MDP policy  $\pi_{\delta,\mathcal{M}}^{\star}$,  suppose $\sum_{t=1}^TR_{P,a}^{\mathcal{M}}(t)/T=\Upsilon_a$ is satisfied. Then under the POMDP policy $\pi_{\delta}^{\star}$,  from Lemma~\ref{lemma: state relationship}, we find that (\ref{PU_IMMR_b}) is equal to (\ref{PU_MDP_ImmR}) under $\pi_{\delta,\mathcal{M}}^{\star}$, i.e.,  $E_{\pi_{\delta}^{\star}}\big\{R_{P,a}(t)|\mathbf{\Lambda}(1)\big\} = R_{P,a}^{\mathcal{M}}(t)$ in each slot $t$. As a result, by summing the PU's  throughput over all $T$ slots, we have  $E_{\pi_{\delta}^{\star}}\big\{\sum_{t=1}^{T} R_{P,a}(t)|\mathbf{\Lambda}(1)\big\}=\Upsilon_a$.
Proposition \ref{proposition: optimal transformed MDP}  thus follows.

\section{Proof of Proposition \ref{proposition: Upper bound}} \label{appendix:proof 6}
The proof is based on  mathematical induction and by computing backward in time.
It is easy to obtain (\ref{upper}) holds in slot $t=T$. Now suppose (\ref{upper})  holds in slot $t+1\leq T$ with  $m_j(t+1)$, $j \in \mathbb{C}_S$, given in (\ref{para_not_last_T}), if $t<T-1$, or  in (\ref{para_last_T}), if $t=T-1$.
Then the following inequality must hold, otherwise,  $\delta_{\mathcal{M}}(t+1)$ can be shown to be negative:
\begin{align*}
 X_P(t\!+\!1)\! \leq  &(\omega_{0}(t\!+\!1)\!+\!\omega_{2}(t\!+\!1))\!\times\! m_1(t\!+\!1) \\
 &\!+\!\omega_{1}(t\!+\!1)\!\times\! m_2(t\!+\!1)\!+\!\omega_{3}(t\!+\!1)\!\times\! m_3(t\!+\!1).
\end{align*}
By substituting (\ref{MDP_TP}) and (\ref{eq: X_P_t}) into the above inequality,
after some algebra, we find that (\ref{upper}) still holds in slot $t$ with $m_j(t)$, $j\in \mathbb{C}_S$,
given  in (\ref{para_not_last_T}).
Proposition \ref{proposition: Upper bound} is thus  proved.

\end{document}